\documentclass{jfm}
\usepackage{epsfig,amssymb,amsfonts,amsmath,psfrag,natbib,version}

\numberwithin{equation}{section}

\newenvironment{remark}{\textit{Remark:}}{}

\newtheorem{thm}{Theorem}                                
 
\newtheorem{lemma}[thm]{Lemma}

\newcommand{\e}{\mathrm{e}}    
\newcommand{\ii}{{\mathrm{i}}}        		

\newcommand{\hlf}{\tfrac{1}{2}}
\newcommand{\half}{\frac{1}{2}}

\newcommand{\rtarr}{\rightarrow}

\newcommand{\del}{\mbox{\boldmath $\nabla$\unboldmath}}
\newcommand{\lap}{\nabla^2}

\newcommand{\reals}{\ensuremath{\mathbb{R}}}
\newcommand{\bigO}{\ensuremath{{\mathcal{O}}}}

\newcommand{\Real}[1]{\text{Re}[#1]}

\newcommand{\Ltwo}{\mathrm{L}^2}

\newcommand{\vecfield}[1]{\mathbf{#1}}

\newcommand{\bu}{\mathbf{u}}
\newcommand{\bv}{\mathbf{v}}
\newcommand{\bU}{\mathbf{U}}
\newcommand{\bk}{\mathbf{k}}
\newcommand{\bx}{\mathbf{x}}

\newcommand{\bJ}{\mathbf{J}}
\newcommand{\btau}{\mathbf{\tau}}
\newcommand{\bO}{\mathbf{0}}

\newcommand{\conj}[1]{{#1}^*}        

\newcommand{\Tub}{A_u^*}
\newcommand{\Tlb}{A_l^*}
\newcommand{\Tt}{T_t^*}
\newcommand{\Tb}{T_b^*}
\newcommand{\Tr}{T_{\text{ref}}}
  
\newcommand{\dT}{\Delta T}    
\newcommand{\dTf}{\Delta T^*}
\newcommand{\dtheta}{\Delta \theta}
\newcommand{\dtau}{\Delta \tau}
\newcommand{\Tscal}{\Theta}

\newcommand{\fluxdim}{\Phi}

\newcommand{\dzT}{\beta}
\newcommand{\dztau}{\gamma}

\newcommand{\kf}{\kappa_f}

\newcommand{\nuf}{\nu_f}        

\newcommand{\rhof}{\rho_f}

\newcommand{\cpf}{c_{p,f}}      
      
\newcommand{\kcof}{\lambda_f}    

\newcommand{\lameig}{\lambda}  
\newcommand{\muwt}{\mu}

\newcommand{\Ra}{\mathit{Ra}}
\newcommand{\Nu}{\mathit{Nu}}
\newcommand{\Prandtl}{\mathit{Pr}}
\newcommand{\R}{R}
\newcommand{\Reff}{R_{\mathrm{e}}}
\newcommand{\Bi}{\eta}

\newcommand{\unitz}{{\bf e_z}}
\newcommand{\unitn}{{\bf n}}

\newcommand{\wh}{\hat{w}}
\newcommand{\whk}{\hat{w}_{\bk}}

\newcommand{\thk}{\hat{\theta}_{\bk}}

\newcommand{\taud}{\tau_{\delta}}

\newcommand{\taubfunc}{[\tau;b]}
\newcommand{\taudbfunc}{[\taud;b]}

\newcommand{\Qcal}{\mathcal{Q}}

\newcommand{\Qcalt}{\Qcal'}

\newcommand{\Qtb}{\Qcal_{\tau,\Reff}}

\newcommand{\Qtiltb}{\Qcalt_{\tau,\Reff}}

\newcommand{\Qk}{\Qcal_{\bk}}
\newcommand{\Qkt}{\Qcal'_{\bk}}

\newcommand{\Qktb}{\Qcal_{\bk; \tau, \Reff}}
\newcommand{\Qkttb}{\Qcal'_{\bk; \tau, \Reff}}

\newcommand{\ELop}{\mathcal{L}_{\tau,\Reff}}

\newcommand{\BdzT}{\mathcal{B}}
\newcommand{\BdT}{\mathcal{D}}

\newcommand{\NuboundR}{\tilde{\mathcal{N}}(\R)}

\newcommand{\Nubnd}{\mathcal{N}}

\newcommand{\BpwldzT}[1]{\mathcal{B}_{\mathrm{pwl},#1}}
\newcommand{\BpwldT}[1]{\mathcal{D}_{\mathrm{pwl},#1}}
\newcommand{\Nupwl}[1]{\mathcal{N}_{\mathrm{pwl},#1}}

\newcommand{\cinf}{c_{\infty}}
\newcommand{\cBi}{c_{\Bi}}
\newcommand{\dsuff}{\delta_{\mathrm{s}}}
\newcommand{\bproof}{b_{\mathrm{c}}}
\newcommand{\badm}{b_{\mathrm{s}}}
\newcommand{\Rsuff}{\R_{\mathrm{s}}}
\newcommand{\Reffs}{\tilde{\Reff}}

\newcommand{\aveT}{T_{a}}

\newcommand{\avetau}{\tau_{a}}

\newcommand{\pdt}[1]{\frac{\partial #1}{\partial t}}

\newcommand{\have}[1]{\overline{#1}}
\newcommand{\tave}[1]{\langle #1 \rangle}
\newcommand{\ltave}[1]{\left\langle #1 \right\rangle}

\newcommand{\zlim}[1]{\left. #1 \right|_{z=0}^1}

\newcommand{\Ltf}[1]{\left\| #1 \right\|}

\newcommand{\intf}{\int} 
\newcommand{\intdf}[1]{\int_{{\mathcal{I}}} #1 \, dS}

\newcommand{\intO}{\int_{\Omega}} 
\newcommand{\intdO}{\int_{\partial \Omega}}

\newcommand{\D}{\mathrm{D}}      

\begin{document}

\title[Bounds for convection with fixed Biot number boundaries]{Bounds
  on Rayleigh-B\'enard convection with general thermal boundary
  conditions. Part 1. Fixed Biot number boundaries}


\author[Ralf Wittenberg]{Ralf W. Wittenberg}
\affiliation{Department of Mathematics, Simon Fraser University,
  Burnaby, BC V5A 1S6, Canada}
\date{\today}

\maketitle

\begin{abstract}
  We investigate the influence of the thermal properties of the
  boundaries in turbulent Rayleigh-B\'enard convection on analytical
  bounds on convective heat transport.  Using the Doering-Constantin
  background flow method, we systematically formulate a bounding
  principle on the Nusselt-Rayleigh number relationship for general
  mixed thermal boundary conditions of constant Biot number $\Bi$
  which continuously interpolates between the previously studied fixed
  temperature ($\Bi = 0$) and fixed flux ($\Bi = \infty$) cases, and
  derive explicit asymptotic and rigorous bounds.  Introducing a
  control parameter $\R$ as a measure of the driving which is in
  general different from the usual Rayleigh number $\Ra$, we find that
  for each $\Bi > 0$, as $\R$ increases the bound on the Nusselt
  number $\Nu$ approaches that for the fixed flux problem.
  Specifically, for $0 < \Bi \leq \infty$ and for sufficiently large
  $\R$ ($\R > \Rsuff = \bigO(\Bi^{-2})$ for small $\Bi$) the Nusselt
  number is bounded as $\Nu \leq c(\Bi) \R^{1/3} \leq C \Ra^{1/2}$,
  where $C$ is an $\Bi$-independent constant.  In the $\R \to \infty$
  limit, the usual fixed temperature assumption is thus a singular
  limit of this general bounding problem.
\end{abstract}


\section{Introduction}
\label{sec:intro}

Rayleigh-B\'enard convection, in which a fluid layer sandwiched
between two horizontal plates is heated from below, has long attracted
considerable experimental and theoretical attention; this is due not
only to its importance as a model with numerous applications in
engineering, geophysics, astrophysics and elsewhere, but also because
it has proved such fertile ground for explorations concerning
stability and dynamics, pattern formation and---under sufficient
heating---convective turbulence (\cite{CrHo93,Kada01}).  Particular
attention has been given to the Nusselt number $\Nu$, representing the
convective enhancement of vertical heat transport, and its dependence
on the Rayleigh number $\Ra$, a measure of the driving via the
temperature difference across the fluid layer.  This dependence
appears to take a scaling form $\Nu \sim f(\Prandtl, \Gamma) \Ra^p$
(with possible logarithmic corrections), where $\Gamma$ represents
geometric effects and $\Prandtl$ is the Prandtl number, and a major
goal of theory and experiment is to find $p$.

Observations that the heat transfer is essentially confined to thermal
boundary layers near the plates, separated by an isothermal core,
suggest that $p = 1/3$, as proposed already by \cite{Malk54}; this
prediction appears to be consistent with large-$\Ra$ experiments
(\cite{NiSr06}) and numerical simulations of turbulent
Boussinesq convection (\cite{AKMSV05}).  In some experiments
(\cite{CCCHCC97,CCCCH01}) an increase in the scaling exponent was
observed at the highest accessible $\Ra$ values, suggesting a possible
transition to a $p = 1/2$ asymptotic regime predicted by
\cite{Krai62}, but other experiments at comparable $\Ra$ have failed
to observe such a transition (\cite{GSNS99,Somm99,NSSD00}), and
possible Prandtl number variabilities or non-Boussinesq effects may
have played a role (\cite{GSNS99,NiSr06b,NiSr06}); strong evidence of
this so-called ``ultimate'' regime appears to be absent.

\subsubsection{Analytical bounds on heat transport:}
\label{sssec:introanalytical}

While good agreement with experimentally observed scaling for a range
of Rayleigh and Prandtl numbers has been attained by a
phenomenological theory (\cite{GrLo01}), we focus on
results derived systematically from the underlying governing
differential equations.  Upper bounding principles derived under some
statistical assumptions, dating to the work of \cite{Howa63} and
\cite{Buss69}, have yielded Kraichan's exponent $p = 1/2$.  More
recently, \cite{DoCo92,DoCo96} realized that an idea of decomposing
flow variables into ``background'' and ``fluctuating'' components,
introduced by \cite{Hopf41} to prove energy boundedness, could be
extended to obtain rigorous analytical bounds (without additional
assumptions) on bulk transport quantities.  The Doering-Constantin
``background flow'' method has since proved remarkably fruitful in
obtaining bounds in good quantitative agreement with experiment or
direct numerical simulation in a wide range of flows.

For Rayleigh-B\'enard convection with fixed temperatures at the lower
and upper boundaries of the fluid, the background flow method yields a
rigorous bound $\Nu \leq C_0 \, \Ra^{1/2}$ uniform in Prandtl number
$\Prandtl$ (\cite{DoCo96}), and while extensive subsequent
investigations (\cite{Kers97,Kers01,PlKe03}) have improved and
optimized the constant $C_0$ in the bound, for general $\Prandtl$ it
has to date only proved possible to lower the exponent $p$ from the
Kraichnan value $p = 1/2$ under additional length scale or regularity
assumptions (\cite{CoDo96,Kers01}).  The assumption of infinite
Prandtl number, however, by imposing an additional constraint on the
velocity field, permits a lowering of the scaling exponent to $p =
1/3$ (with possible logarithmic corrections)
(\cite{Chan71,CoDo99,IKP06}); the best current rigorous bound in this
case has the form $\Nu \leq C (\ln \Ra)^{1/3} \Ra^{1/3}$
(\cite{DOR06}), and related results have recently been obtained by
\cite{Wang08b} for sufficiently large finite Prandtl number.

\subsubsection{Influence of thermal properties of the plates:}
\label{sssec:introthermal}

The above analyses were performed under the usual assumption that the
lower and upper boundaries of the fluid in Rayleigh-B\'enard
convection are held at known uniform temperature, or equivalently,
that the bounding plates are perfect conductors.  In practise, though,
the boundaries are imperfectly conducting; and the thermal properties
of the boundaries have long been understood to affect the initial
instability to convection and the weakly nonlinear behaviour beyond
transition (see for instance \cite{SGJ64,HJP67,ChPr80,BuRi80} and the
review in \cite{CrHo93}).  Even when the bounding plates have much
higher conductivity than the fluid, as the Rayleigh number (and hence
the Nusselt number) increases, the \emph{effective} conductivity of
the fluid, depending on $\Nu$, eventually becomes comparable to and
then exceeds that of the plates; in fact, in the $\Ra \to \infty$
limit one might expect the fluid effectively to ``short circuit'' the
system, with the bounding plates acting essentially as perfect
insulators by comparison.

The effect of the finite thermal conductivity of the bounding plates
on convective heat transport has stimulated recent modelling
(\cite{CCC02,CRCC04}), experimental (\cite{BNFA05}) and numerical
(\cite{Verz04}) studies with the aim of reconciling various
experimental results with each other and with fixed-temperature
theoretical (\cite{GrLo01}) and numerical (\cite{AKMSV05})
predictions.  Recent numerical simulations comparing fixed flux and
fixed temperature boundary conditions have reached differing
conclusions: The computations of \cite{VeSr08} in cylindrical geometry
found $\Nu \sim \Ra^{1/3}$ scaling, but that for a given large enough
$\Ra$, the Nusselt number is reduced upon replacement of fixed
temperature conditions at the lower boundary of the fluid by fixed
flux conditions.  \cite{JoDo07,JoDo08}, on the other hand, found that
in numerical integration of two-dimensional, horizontally periodic
convection, the heat transport for large $\Ra$ was the same, namely
$\Nu \sim \Ra^{2/7}$, for fixed temperature and fixed flux conditions
at the upper and lower boundaries of the fluid.

Predating most of the above recent investigations and with similar
motivations, \cite{OWWD02} initiated the analytical study of the
effects of thermal boundary conditions on Rayleigh-B\'enard convection
using the background flow method, obtaining a bound on the heat
transport with fixed flux boundary conditions at the fluid boundaries,
which again took the form $\Nu \leq C_{\infty} \Ra^{1/2}$ (this work
was recently extended to porous medium convection by \cite{Wei07}).
However, the mathematical structure of the fixed flux bounding problem
of \cite{OWWD02}, and various intermediate scaling results, turned out
to be quite different from the fixed temperature case
(\cite{DoCo96,Kers01}).  It is thus natural to wonder how these two
extreme cases, corresponding respectively to the idealizations of
perfectly conducting and insulating plates, are related
\textit{vis-\`a-vis} their bounding problems, and which is more
relevant to real, finitely conducting boundaries.

\subsubsection{Outline of this paper:}
\label{sssec:introoutline}

In the present work we reconsider the effect of general thermal
boundary conditions on systematically derived analytical bounds on
thermal convection, continuing the program initiated by \cite{OWWD02};
for simplicity we consider only identical thermal properties at the
top and bottom fluid boundaries in the mathematically idealized
horizontally periodic case.  We consider a common model for poorly
conducting plates, namely mixed (Robin) thermal boundary conditions of
``Newton's Law of Heating'' type, with a fixed Biot number $\Bi$, so
that $\Bi = 0$ gives the fixed temperature and $\Bi = \infty$ the
fixed flux case; to our knowledge the only prior bounding study with
general Biot number is the work of \cite{SKB04} on horizontal
convection, in which mixed thermal conditions were imposed at the
lower boundary of the fluid.

In Section~\ref{sec:fluidpdes} of this paper, we carefully develop a
general formulation and bounding principle using the
Doering-Constantin background flow method, and in
Section~\ref{sec:mixedth} specialize to mixed thermal conditions in a
manner that interpolates smoothly between the fixed temperature and
fixed flux cases.  The use of a piecewise linear background
temperature profile and explicit estimates derived in
Section~\ref{sec:pwlinest} enables us, in Section~\ref{sec:Bibound},
to derive analytical bounds on the $\Nu$--$\Ra$ relationship
asymptotically valid for $\Ra \to \infty$.  For completeness, we also
prove the $\Ltwo$ boundedness of temperature and velocity fields in
Appendix~\ref{app:bounded}, and prove rigorous, though less sharp,
bounds on the $\Nu$--$\Ra$ relationship in
Appendix~\ref{app:rigbounds}.

Summarizing our results: 
Since in general the boundary temperatures are unknown \emph{a
  priori}, it is necessary to introduce a control parameter $\R$,
which equals the standard Rayleigh number $\Ra$ only in the fixed
temperature case $\Bi = 0$.  \cite{OWWD02} showed in the
fixed flux case $\Bi = \infty$ that while the bound was $\Nu \leq
C_{\infty} \Ra^{1/2}$, it was obtained through the estimates $\Nu \leq
c_1 \R^{1/3}$, $\Ra \geq c_2 \R^{2/3}$, quite unlike the fixed
temperature case $\Nu \leq C_0 \R^{1/2}$, $\Ra = \R$.  

In the present work, for general sufficiently small Biot number $\Bi$
we show that for small $\R$ we have $\Nu \lesssim \bigO(\R^{1/2})$,
$\Ra \gtrsim \bigO(\R)$ as in the fixed temperature case, but for $\R$
(and hence $\Ra$) beyond some critical parameter $\Rsuff =
\bigO(\Bi^{-2})$, we find $\Nu \leq c_1(\Bi) \R^{1/3}$, $\Ra \geq
c_2(\Bi) \R^{2/3}$, implying $\Nu \leq C_{\Bi} \Ra^{1/2}$ with
intermediate scaling as in the fixed flux case.  Interestingly, for
$\Bi > 0$ we find $C_{\Bi} = C_{\infty}$: at least at the level of our
estimates, the asymptotic scaling in each case is as for fixed flux
boundary conditions (providing rigorous support for the intuition that
for sufficiently high $\Ra$, the plates essentially act as
insulators), while the fixed temperature problem is a singular limit
of the general asymptotic bounding problem.  More details of the
scaling of the bounds in different regimes, together with numerically
obtained conservative bounds for piecewise linear backgrounds, will be
given elsewhere (\cite{WiGa08prep}).

The use of mixed ``Newton's Law of Cooling'' boundary conditions with
fixed Biot number $\Bi$ to model imperfectly conducting boundaries is
a simplification, however, since in general the Biot number depends on
horizontal wave number (see for instance \cite{CrHo93}).  In a
subsequent paper (\cite{Witt08ub}, Part~2 of the present work), we
improve upon our model by formulating and obtaining bounds for the
more realistic problem of a fluid bounded by plates of finite
thickness and conductivity, establishing a systematic correspondence
between that situation and the present fixed Biot number case.

\section{Governing equations and bounding principle
  with general thermal conditions at fluid boundaries}
\label{sec:fluidpdes}

We begin by formulating the standard Rayleigh-B\'enard convection
problem in the fluid and developing a bounding principle in the usual
way, but without fixing the thermal conditions at the fluid
boundaries.  For reference and clarity, though, we occasionally point
out the forms of our results in the fixed temperature and fixed flux
special cases previously treated in the literature, as our development
is designed to interpolate between these extremes.

\subsection{Governing differential equations and nondimensionalization}
\label{ssec:Bousnondim}

In the Boussinesq approximation, the equations of motion in the fluid
are 
\begin{align}
  \label{eq:Bousdim1}
  \frac{\partial \bu^*}{\partial t^*} + \bu^* \cdot \del^* \bu^* +
  \frac{1}{\rhof} \del^* P^* & = \nuf 
  \nabla^{*2} \bu^* + \alpha g (T^* - T_0) \unitz \ , \\
  \label{eq:Bousdim2}
  \del^* \cdot \bu^* & = 0 \ , \\
  \label{eq:Bousdim3}
  \frac{\partial T^*}{\partial t^*} + \bu^* \cdot \del^* T^* & = \kf
  \nabla^{*2} T^* \ , \\ 
  \label{eq:Bousdim4}
  \bu^*|_{z^* = 0,h} & = \bO \ ,
\end{align}
where $\nuf$ and $\kf$ are the momentum and thermal diffusivities of
the fluid, respectively, $g$ is the acceleration due to gravity,
$\alpha$ is the thermal expansion coefficient, $\rhof$ the fluid
density at some reference temperature $T_0$, and $h$ is the height of
the fluid layer, as in Figure~\ref{fig:RBsetup}; we also let $\cpf$ be
the specific heat and $\kcof = \rhof \, \cpf \, \kf$ be the thermal
conductivity of the fluid.  In this formulation, the compressibility
of the fluid is neglected everywhere except in the buoyancy force
term, and the pressure is determined via the divergence-free condition
on $\bu^*$.  Variables with an asterisk are dimensional, we have
no-slip velocity boundary conditions in the vertical direction, and we
take periodic boundary conditions in the horizontal directions, with
periods $L^*_x$ and $L^*_y$, respectively.
\begin{figure}
  \begin{center}
        \includegraphics[width = 3.5in]{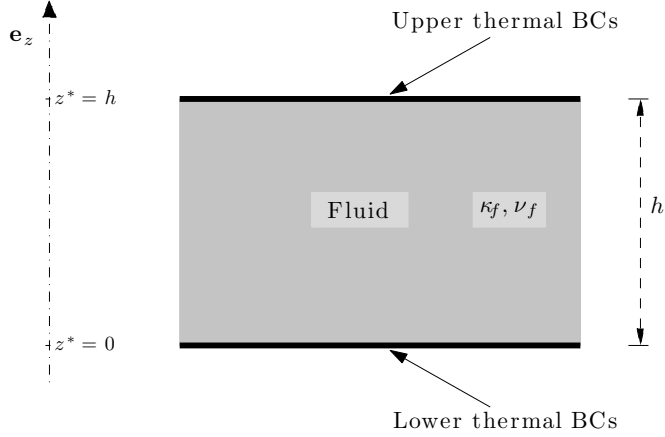}
    \caption{Geometry of Rayleigh-B\'enard convection system with
      thermal boundary conditions imposed at upper and lower limits of
      fluid layer.}
    \label{fig:RBsetup}
  \end{center}
\end{figure}

The present formulation is designed to be flexible with respect to the
choice of thermal boundary conditions (BCs) at the plate-fluid
interfaces $z^* = 0$ and $z^* = h$; so in general we let $\Tscal$ be a
given temperature scale, determined according to the thermal BCs, and
introduce a reference (``zero'') temperature $\Tr$.

The usual and most-studied assumption regarding thermal boundary
conditions at the interfaces is that the \emph{temperature is fixed} at
the upper and lower boundaries:
\begin{equation}
  \label{eq:dirbcdim}
  T^*|_{z^*=0} \equiv T^*(x^*,y^*,0,t^*) = \Tb, \qquad T^*|_{z^* = h} = \Tt.
\end{equation}
These Dirichlet BCs imply a natural choice of reference temperature
$\Tr = \Tt$, while the imposed temperature drop $\Delta T^* \equiv
-\left. T^* \right|_{z^*=0}^{h} \equiv \Tb - \Tt$ introduces a natural
temperature scale $\Tscal = \Delta T^*$.  At the opposite extreme is
the \emph{fixed flux assumption} that the thermal heat
flux $-\kcof T^*_{z^*} \equiv - \kcof \, \partial T^*/\partial z^*$
through the fluid boundaries is a constant, which we call $\fluxdim$.
This corresponds to the Neumann BCs of fixed normal temperature
gradient $-\dzT^*$ at the interfaces:
\begin{equation}
  \label{eq:neubcdim}
  T^*_{z^*}|_{z^* = 0} = T^*_{z^*}|_{z^* = h} = - \dzT^* = -
  \frac{\fluxdim}{\kcof}  ;
\end{equation}
the corresponding temperature scale is $\Tscal = h \dzT^* = h
\fluxdim/\kcof$, while in this case $\Tr$ is arbitrary.
More general mixed (Robin) thermal BCs are discussed later in
Section~\ref{sec:mixedth}. 

We nondimensionalize using $\Tr$ and the temperature scale $\Tscal$,
and with respect to the fluid layer thickness $h$ and thermal
diffusivity time $h^2/\kf$; that is, we take $h$, $h^2/\kf$,
$U = \kf/h$ and $\rhof U^2$ as our appropriate length, time, velocity
and pressure scales respectively.  For $\Tr \not= T_0$, the
nondimensional fluid momentum equation will contain a constant term
proportional to $\Tr - T_0$ in the $\unitz$ direction, which we can
take care of by absorbing it into the rescaled pressure; this effect
of the buoyancy force corresponds to a linear vertical pressure
gradient.
In summary, the nondimensional variables (without asterisks) are
defined by:
\begin{equation*}
  \bx = \frac{\bx^*}{h}, \quad t = \frac{t^*}{t_{\text{scal}}}, \quad 
  \bu = \frac{\bu^*}{U}, \quad T = \frac{T^* - \Tr}{\Tscal},
  \quad  p  = \frac{1}{\Prandtl} \frac{P}{\rhof U^2} -
  \R \frac{\Tr - T_0}{\Tscal} z,
\end{equation*}
where $t_{\text{scal}} = h^2/\kf$, $U = \kf/h$, and $\Prandtl$ and
$\R$ are defined below.  The dimensionless periodicity lengths in the
transverse directions are $L_x = L^*_x/h$ and $L_y = L^*_y/h$.

The equations for the nondimensional fluid velocity $\bu = (u,v,w)$
and fluid temperature $T$ are thus
\begin{align}
  \label{eq:Bous}
  \Prandtl^{-1} \left(\pdt{\bu} + \bu \cdot \del \bu \right) +
  \del p & = \lap \bu + \R \, T \, \unitz , \\
  \del \cdot \bu & = 0 ,\label{eq:divfree} \\
  \pdt{T} + \bu \cdot \del T & = \lap T , \label{eq:heatfluid}
\end{align}
with no-slip BCs $\bu|_{z = 0,1} = \bO$, and $L_x,\ L_y$-periodic BCs
in the horizontal $x$ and $y$ directions in all variables.  Here the
usual Prandtl number $\Prandtl$ and the \emph{control parameter} $\R$
are defined by
\begin{equation}
  \label{eq:PrR}
  \Prandtl = \frac{\nuf}{\kf}, \quad \R = \frac{\alpha g h^3}{\nuf \kf}
  \Tscal.
\end{equation}
In the fixed temperature (Dirichlet) and fixed flux (Neumann) limits,
the nondimensional thermal BCs are, respectively,
\begin{eqnarray}
  \label{eq:dirbc}
  T|_{z=0} = 1, \qquad T|_{z=1} = 0 , \\
  \label{eq:neubc}
  \left. T_z \right|_{z=0} = \left. T_z \right|_{z=1}  = -1 .
\end{eqnarray}

\subsubsection{Some additional notation:}
\label{sssec:notation}

Following \cite{OWWD02}, for functions $h(x,y,z)$
and $g(t)$ we define the horizontal and time averages, $\have{h}(z)$
and $\tave{g}$ respectively, by 
\begin{equation*}
  \have{h}(z) = \frac{1}{A} \iint_A h(x,y,z) \, dx dy = \frac{1}{A}
  \int_0^{L_y} \int_0^{L_x} h(x,y,z)\, dx dy 
\end{equation*}
and
\begin{equation*}
  \tave{g} = \limsup_{\tau \rtarr \infty} \frac{1}{\tau} \int_0^{\tau}
  g(t)\,dt ,
\end{equation*}
where $A = L_x L_y$ is the nondimensional area of the plates.  Also,
$\intf h = A \int_0^1 \have{h}(z) \, dz$ denotes a volume integral
over the entire fluid layer, and the $\Ltwo$ norms are defined over
the fluid by
\begin{equation*}
  \Ltf{h}^2 = \intf h^2 = \int_0^1 \iint_A h^2(x,y,z) \, dxdy\, dz.
\end{equation*}
Finally, over the domain we consider with horizontally periodic
BCs, surface integrals of vector fields $\vecfield{H}$
over the fluid boundary (the interfaces ${\mathcal{I}}$ between the
plates and the fluid) are
\begin{equation*}
  \intdf{\unitn \cdot \vecfield{H}} = A \zlim{\have{\vecfield{H} \cdot
  \unitz}}.
\end{equation*}

\subsection{Rayleigh and Nusselt numbers}
\label{ssec:RaNu}

We define the nondimensional horizontally- and time-averaged
temperature drop across the fluid as
\begin{equation}
  \label{eq:deltadef}
  \dT = - \tave{\zlim{\have{T}}} = \tave{\have{T}|_{z=0} - \have{T}|_{z=1}} =
  \frac{\dTf}{\Tscal} ,
\end{equation}
where $\dTf = \tave{\have{T^*}|_{z^* = 0} - \have{T^*}|_{z^* = h}}$; 
we observe that this is known \emph{a priori} only for fixed
temperature BCs, in which case $\dTf = \Tb - \Tt =
\Tscal$, $\dT = 1$.
The conventional Rayleigh number $\Ra$ is defined in terms of this
averaged temperature difference as 
\begin{equation}
  \label{eq:Radef}
    \Ra  = \frac{\alpha g h^3}{\nuf \kf}  \dTf =
    \frac{\alpha g h^3 \Tscal}{\nuf \kf} \dT ,
\end{equation}
showing that $\Ra$ is related to the control parameter $\R$ by
\begin{equation}
  \label{eq:RRa}
  \Ra = \R \, \dT .
\end{equation}

The Nusselt number $\Nu$ is a nondimensional measure of the total heat
transport through the fluid layer, which may be derived by first
rewriting the thermal advection equation in the fluid
\eqref{eq:heatfluid} as a conservation law, $T_t + \del \cdot \bJ = 0$
(using incompressibility).  Here the dimensionless heat current $\bJ =
\bu T + \bJ_c$ is composed of the conductive heat current $\bJ_c = -
\del T$ and the convective heat current $\bJ_v = \bu T$, with
corresponding overall instantaneous conductive and convective vertical
heat transport $j_c(t) = \intf \unitz \cdot \bJ_c = \iint_A \int_0^1 -
T_z \, dz \, dx dy = - A \zlim{\have{T}}$ and $j_v(t) = \intf \unitz
\cdot \bJ_v = \intf w T$, respectively.
The Nusselt number $\Nu$ is now defined as the ratio of the total
(averaged) vertical heat transport, $j(t) =
\intf \unitz \cdot \bJ = j_c(t) + j_v(t)$, to the purely conductive
transport:
\begin{equation}
  \label{eq:Nudef}
  \Nu = \frac{\tave{j_c(t) + j_v(t)}}{\tave{j_c(t)}} = 1 +
  \frac{\frac{1}{A}\tave{\intf w T}}{\dT} .
\end{equation}

A more useful expression, which allows us to estimate $\Nu$ from the
equations of motion, is found by relating $\tave{\intf w T}$ to the
time-averaged temperature drop and boundary flux.  To do so, we begin
by taking the horizontal average of the temperature equation
\eqref{eq:heatfluid}, using the horizontally periodic BCs, to get
\begin{equation}
  \label{eq:tfhave}
  \pdt{\have{T}} + \have{\del \cdot \bJ} = \pdt{\have{T}} +
  \frac{\partial}{\partial z} \left(\have{ w T} - \have{T}_z
  \right) = 0 .
\end{equation}
Integrating over $z$ and using the vertical no-slip boundary
conditions on $w$, 
\begin{equation}
  \label{eq:tfave}
  \frac{d}{dt} {\intf T} + A \int_0^1 \left( \have{wT} - \have{T}_z
  \right)_z  \, dz = \frac{d}{dt} {\intf T} + A
  \zlim{(-\have{T}_z)} = 0. 
\end{equation}
As we show in Appendix~\ref{app:bounded}, the thermal energy
$\Ltf{T}^2 = \intf T^2$ is uniformly bounded in time for the thermal
BCs we consider, so that $\intf T$ is also uniformly bounded by the
Cauchy-Schwarz lemma via $\intf T \leq A^{1/2} \Ltf{T}$.  Hence on
taking a time average of \eqref{eq:tfave}, the time derivative term
vanishes, and we find that
\begin{equation}
  \label{eq:betadef}
  \left. \ltave{- \have{T}_z} \right|_{z=0} = 
  \left. \ltave{- \have{T}_z} \right|_{z=1} \equiv \dzT ,
\end{equation}
where the above expression defines $\dzT$, the horizontally- and
time-averaged vertical temperature gradient, or equivalently, the
\emph{nondimensional heat flux} at the fluid boundaries; note that
(only) in the fixed flux case, this quantity is known, $\dzT = 1$.  As
expressed by \eqref{eq:betadef}, obviously on average, there is a
balance between the heat fluxes entering the fluid layer at the bottom
and leaving at the top.

Taking a time average of \eqref{eq:tfhave}, via a maximum principle
on $T$ the time derivative term $\tave{\partial \have{T}/\partial t}$
would vanish; however, for fixed flux BCs we do not have such an
\emph{a priori} maximum principle.  Following \cite{OWWD02}, uniformly
in thermal BCs we can instead multiply \eqref{eq:tfhave} by $z$ and
integrate to obtain
\begin{equation}
  \label{eq:tfzave}
  \frac{d}{dt} \int_0^1 z \have{T}\,dz + \int_0^1 z \left( \have{wT} -
    \have{T}_z \right)_z \, dz  = 0 \ ; 
\end{equation}
integrating the second term by parts and using the no-slip boundary
conditions, we have 
\begin{equation}
  \label{eq:tfzave2}
   \int_0^1 z \left( \have{wT} - \have{T}_z \right)_z \, dz = z
   \left. \left( \have{wT} - \have{T}_z\right) \right|_0^1 - 
    \int_0^1 \left( \have{wT} - \have{T}_z \right) \, dz 
    = - \have{T}_z|_{z=1} - \frac{1}{A} \intf w T + \zlim{\have{T}} .
\end{equation}
Now as before, via 
\begin{equation*}
  \int_0^1 z \have{T}\,dz = \frac{1}{A} \intf z T \leq \frac{1}{A}
  \left(\intf z^2\right)^{1/2} \left( \intf T^2 \right)^{1/2} = (3
  A)^{-1/2} \Ltf{T}
\end{equation*}
and the uniform boundedness of $\Ltf{T}$, the time average of the
first term in \eqref{eq:tfzave} vanishes.  Thus taking time averages
of \eqref{eq:tfzave}, using \eqref{eq:tfzave2} and the definitions
\eqref{eq:deltadef} and \eqref{eq:betadef}, we obtain
\begin{equation}
  \label{eq:wTf}
  \frac{1}{A} \ltave{\intf w T} = \dzT - \dT .
\end{equation}

The expression \eqref{eq:wTf} may now immediately be substituted into
\eqref{eq:Nudef} to obtain the fundamental identity for the Nusselt
number for general thermal BCs,
\begin{equation}
  \label{eq:Nueq}
  \Nu = \frac{\dzT}{\dT} ,
\end{equation}
and consequently, using \eqref{eq:RRa}, we have $\Nu \, \Ra = \R \,
\dzT$.

In the special case of fixed temperature BCs, the identity
\eqref{eq:Nueq} is well-known (\cite{DoCo96}): with $\dT = 1$, the
Nusselt number equals the time- and horizontally-averaged flux at the
boundary, $\Nu = \dzT = \tave{- \have{T}_z|_{z=0,1}}$, and thus an
upper bound on $\Nu$ is found by bounding $\dzT$ from above.  For the
fixed flux case, $\dzT = 1$ is known, the Nusselt number is the
inverse of the averaged temperature drop, $\Nu = \dT^{-1}$
(\cite{OWWD02}), and one seeks a \emph{lower} bound on $\dT$.  In
general, though, neither $\dT$ nor $\dzT$ is known \emph{a priori},
though they are related through the thermal BCs, as discussed in Section~\ref{sec:mixedth}.

\subsection{Global energy balance}
\label{ssec:globalenergy}

We next obtain the basic ``energy'' identities from the
governing Boussinesq equations, which allow us to relate the Nusselt
number to the momentum and heat dissipation: 
First, taking the inner product of
the momentum equation \eqref{eq:Bous} with $\bu$, integrating over the fluid
domain, integrating by parts and using incompressibility and no-slip
BCs, we find 
\begin{equation}
  \label{eq:fluiden}
  \frac{1}{2 \Prandtl} \frac{d}{dt} \Ltf{\bu}^2 = - \Ltf{\del \bu}^2
  + R \intf w T . 
\end{equation}
The $\Ltwo$ norm of the fluid velocity is \emph{a priori} bounded, as
shown in Appendix~\ref{app:bounded}; hence, taking time averages, we
derive the result (using \eqref{eq:wTf})
\begin{equation}
  \label{eq:gradu2}
  \frac{1}{R} \tave{\Ltf{\del \bu}^2 } = \ltave{\intf w T} = A \left(
    \dzT -  \dT \right).
\end{equation}
Observe that \eqref{eq:gradu2} implies that $\dzT \geq \dT$, so that
by \eqref{eq:Nueq} we have $\Nu \geq 1$, as expected.

Similarly, we can take multiply the thermal
advection-diffusion equation \eqref{eq:heatfluid} by $T$, integrate over the
fluid and integrate by parts.  Neither the advection term nor the
horizontal boundary terms contribute, so we find
\begin{equation}
  \label{eq:heaten}
  \frac{1}{2} \frac{d}{dt} \Ltf{T}^2 = - \Ltf{\del T}^2 +
  \intdf{\unitn \cdot (T \del T)} = - \Ltf{\del T}^2 + A
  \zlim{\have{T T_z}}. 
\end{equation}
Taking time averages and using boundedness of $\Ltf{T}$, we find
\begin{equation}
  \label{eq:gradT}
  \tave{\Ltf{\del T}^2 } = A \ltave{\zlim{\have{T T_z}} } .
\end{equation}

Again, for reference we quote these results in the known limits: for
fixed temperature BCs \eqref{eq:dirbc} the global energy identities
\eqref{eq:gradu2}, \eqref{eq:gradT} become (\cite{DoCo96})
\begin{align}
  \frac{1}{A R} \tave{\Ltf{\del \bu}^2 } & = \left( \dzT
  - 1 \right) = \Nu - 1 , \label{eq:gradudir} \\
  \frac{1}{A} \tave{\Ltf{\del T}^2 } & = \tave{ - \have{T_z}(0) } =
  \dzT =  \Nu ; \label{eq:gradTdir}
\end{align}
while in the fixed flux case \eqref{eq:neubc}, they are (\cite{OWWD02})
\begin{align}
  \frac{1}{A R} \tave{\Ltf{\del \bu}^2 } & = \left( 1 - \dT \right)
  = 1 - \Nu^{-1}  , \label{eq:graduneu} \\
  \frac{1}{A} \tave{\Ltf{\del T}^2 } & = \tave{ \zlim{ -
  \have{T}} } = \dT = \Nu^{-1} . \label{eq:gradTneu}
\end{align}

\subsection{Background fields}
\label{ssec:backgroundfields}

In the spirit of the ``background'' method of Doering and Constantin,
we now introduce a decomposition of the flow into a background, which
carries the boundary conditions of the flow, and a space- and
time-dependent fluctuating field (\cite{DoCo96,Kers01}).  One chooses
fields $\bU$ and $\btau$ which satisfy the same velocity and
temperature BCs as $\bu$ and $T$, and appropriate evolution
equations---such $(\bU,\btau)$ are referred to as ``background'' flow
and temperature fields---and lets $\bv(\bx,t)$ and $\theta(\bx,t)$ be
arbitrary space- and time-dependent perturbations satisfying
homogeneous BCs, so that the velocity and temperature fields are
decomposed into a background plus a fluctuation, according to
$\bu(\bx,t) = \bU + \bv(\bx,t)$, $T(\bx,t) = \btau + \theta(\bx,t)$.

For simplicity of presentation we restrict our attention to the case
of zero background velocity field and a $z$-dependent temperature
background, $(\bU,\btau) = (\bO,\tau(z))$, as it appears that more
general backgrounds are unlikely to improve the overall scaling of the
bounds (\cite{Kers01}).  Furthermore, when the upper and lower
boundaries of the fluid have identical thermal properties, it is
sufficient to consider only background fields satisfying $\tau'(0) =
\tau'(1)$ (compare \eqref{eq:betadef}), and we define
\begin{equation}
  \label{eq:gammadtaudef}
  \dtau = \tau(0) - \tau(1), \qquad \dztau = - \tau'(0) = - \tau'(1) .
\end{equation}

We thus define $\bv$ and $\theta$ via the decomposition 
\begin{equation}
  \label{eq:decomp}
  \bu(\bx,t) = \bv(\bx,t), \qquad T(\bx,t) = \tau(z) + \theta(\bx,t) , 
\end{equation}
(note that we prefer to preserve the (notational) distinction between
the overall velocity field $\bu$ and the fluctuating field $\bv$,
though for convenience we ignore this distinction in writing the
components $(u,v,w)$ of the velocity field, for instance $w = \bu
\cdot \unitz = \bv \cdot \unitz$) and immediately obtain identities
between the norms of the gradients of the full solutions and their
fluctuations:
\begin{align}
  \frac{1}{\R} \Ltf{\del \bu}^2 & = \frac{1}{\R} \Ltf{\del \bv}^2
  , \label{eq:graduid} \\ 
  \Ltf{\del T}^2 & = \Ltf{\del \theta}^2 + 2 \intf \theta_z \tau'
  + \intf \tau'^2 . \label{eq:gradTid} 
\end{align}

Inserting the decomposition \eqref{eq:decomp} into the Boussinesq
equations yields the evolution equations for the perturbations,
\begin{align}
  \label{eq:fluctmom}
  \Prandtl^{-1} \left(\pdt{\bv} + \bv \cdot \del \bv \right) +
  \del \tilde{p} & = \lap \bv + \R \, \theta \, \unitz , \\
  \del \cdot \bv & = 0 ,\label{eq:fluctdivfree} \\
  \pdt{\theta} + \bv \cdot \del \theta & = \lap \theta + \tau'' - w
  \tau' , \label{eq:fluctheat}
\end{align}
where $\bv$ and $\theta$ satisfy appropriate homogeneous boundary
conditions; that is, all fields are horizontally periodic, $\bv$
satisfies the no-slip BCs $\bv|_{z=0,1} = \bO$, and $\theta$ satisfies
homogeneous thermal BCs consistent with those for $T$.
In particular, for fixed temperature BCs, we have $\tau(0) = 1$,
$\tau(1) = 0$, and $\theta$ satisfies homogeneous Dirichlet BCs,
$\theta|_{z=0,1} = 0$; while in the fixed flux case, we require
$\tau'(0) = \tau'(1) = -\dztau = -1$, so that $\theta$ satisfies the
homogeneous Neumann BCs $\theta_z|_{z=0,1} = 0$.

In an analogous way to the calculations of
Section~\ref{ssec:globalenergy}, we may find the energy identities for
the fluctuations $\bv$ and $\theta$: 
Taking the inner product of \eqref{eq:fluctmom} with $\bv$,
integrating over the fluid and using incompressibility, the evolution
equation for the $\Ltwo$ norm of the fluctuating velocity field is 
\begin{equation}
  \label{eq:vL2}
  \frac{1}{2 \Prandtl \R} \frac{d}{dt} \Ltf{\bv}^2 = - \frac{1}{\R}
  \Ltf{\del \bv}^2 + \intf w \theta , 
\end{equation}
as in \eqref{eq:fluiden}.  
We find the corresponding $\Ltwo$ evolution for the perturbed temperature
$\theta$ by multiplying \eqref{eq:fluctheat} by $\theta$ and
integrating:
\begin{equation}
  \label{eq:thL2}
  \frac{1}{2} \frac{d}{dt} \Ltf{\theta}^2  = - \Ltf{\del \theta}^2 +
  A \zlim{\have{\theta \theta_z}} - \intf \theta_z
  \tau' + A \zlim{\have{\theta} \tau'} - \intf w \theta \tau' .
\end{equation}
For future reference, using \eqref{eq:gammadtaudef} and the decomposition
\eqref{eq:decomp}, the second boundary term in \eqref{eq:thL2} has
time average
\begin{equation}
  \label{eq:thbdy1}
  \tave{\zlim{\have{\theta} \tau'}} = - \dztau
  \tave{\zlim{\have{\theta}}} \equiv \dztau \dtheta = \dztau \left(
    \dT - \dtau \right) .
\end{equation}

\subsection{Governing equations for general bounding principle}
\label{ssec:genformulation}

In order to formulate upper bounding principles for the Nusselt
number, we take appropriate linear combinations of the $\Ltwo$ identities
\eqref{eq:vL2}--\eqref{eq:thL2} for the fluctuating quantities, and
the identities \eqref{eq:graduid}--\eqref{eq:gradTid} for the
decomposition of the gradient into background and fluctuating parts.
In general, such linear combinations may contain three free parameters,
over which one might optimize to obtain the best possible bound
available within such a formalism (\cite{Kers97,Kers01}).
However, we shall consider only the restricted special case of a
single ``balance parameter'' $b$ (\cite{NGH97}); in this simplified
formulation, in fact we do not require the energy dissipation equation
\eqref{eq:vL2} at all.

We can eliminate the $\intf \theta_z \tau'$ term in the thermal energy
equations by taking $2 \cdot \eqref{eq:thL2} + \eqref{eq:gradTid}$ to
give
\begin{equation}
  \label{eq:thbal1}
  \frac{1}{2} \frac{d}{dt} \left( 2 \Ltf{\theta}^2 \right) +
  \Ltf{\del T}^2 = \intf \tau'^2 - \Ltf{\del \theta}^2 - 2
  \intf \theta w \tau' + 2 A \zlim{\have{\theta \theta_z}} + 2 A
  \zlim{\have{\theta} \tau'} .
\end{equation}
Now taking a weighted average $b \cdot \eqref{eq:thbal1} + (1-b) \cdot
\eqref{eq:graduid} = b \cdot [2\eqref{eq:thL2} + \eqref{eq:gradTid}] +
(1-b) \cdot [0\eqref{eq:vL2} + \eqref{eq:graduid}]$, we have
\begin{equation}
  \label{eq:thvbal1}
  \begin{split}
    \frac{1}{2} \frac{d}{dt} \left( 2 b \Ltf{\theta}^2 \right) + b
    \Ltf{\del T}^2 + \frac{1-b}{R} \Ltf{\del \bu}^2 
    = b \intf \tau'^2 + 2 b A \zlim{\have{\theta} \tau'} \\
    + 2 b A \zlim{\have{\theta \theta_z}} + \intf \left[
      \frac{1-b}{\R} | \del \bv|^2 - 2 b \tau' w
      \theta - b |\del \theta|^2       \right] .
  \end{split}
\end{equation}

We now take time averages and note that the time derivative term
vanishes due to the boundedness of $\Ltf{\theta}$, as shown in
Appendix~\ref{app:bounded}.  Using \eqref{eq:gradu2}, \eqref{eq:gradT}
and \eqref{eq:thbdy1}, we find
\begin{equation}
  \label{eq:bdid1}
  b A \tave{\zlim{\have{T T_z}}} + (1-b) A (\dzT - \dT) = b \intf
  \tau'^2 + 2 b A \dztau \dtheta + 2 b A \tave{\zlim{\have{\theta
  \theta_z}}} -  b \, \Qtb [\bv,\theta] , 
\end{equation}
where we define the quadratic form
\begin{equation}
  \label{eq:Qtbdef}
  \Qtb [\bv,\theta] = \ltave{\intf \left[ \frac{b-1}{b \R} |\del
      \bv|^2 + 2 \tau' w \theta + |\del \theta|^2  \right]} =
  \ltave{\intf \left[ \frac{1}{\Reff} |\del 
      \bv|^2 + 2 \tau' w \theta + |\del \theta|^2  \right]} .
\end{equation}
Here we have defined an \emph{``effective control parameter''} $\Reff$
via
\begin{equation}
  \label{eq:Reffdef}
  \Reff = \frac{b}{b-1} \R ,
\end{equation}
having observed that the quadratic form $\Qtb$ depends on $\R$ and the
balance parameter $b$ only through the combination $b \R/(b-1)$.  We
desire a positive balance parameter $b$ so that a lower bound on
$\Qtb$ (and hence on $\Qtiltb$ in
\eqref{eq:bdid2}--\eqref{eq:Qtiltbdef} below) should imply an upper
bound on $\dzT$ and/or a lower bound on $\dT$; since a necessary
condition for $\Qtb$ to be a positive definite quadratic form is that
$\Reff > 0$ or $(b-1)/b > 0$, we thus require $b > 1$.

Continuing with the formulation of the governing equations, using
\eqref{eq:decomp}, \eqref{eq:betadef} and \eqref{eq:gammadtaudef}, we
decompose the first term in \eqref{eq:bdid1} via
\begin{equation}
  \label{eq:TTzid}
  \tave{\zlim{\have{T T_z}}} = \tave{\zlim{\tau \have{T_z}}} +
  \tave{\zlim{\have{\theta} \tau'}} + \tave{\zlim{\have{\theta
  \theta_z}}} = \dzT \dtau + \dztau \dtheta + \tave{\zlim{\have{\theta
  \theta_z}}}. 
\end{equation}
Substituting \eqref{eq:TTzid} into \eqref{eq:bdid1}, writing $\dtheta
= \dT - \dtau$ and rearranging terms, we obtain
\begin{equation}
  \label{eq:bdid2}
  (1-b) (\dzT - \dT) + b \left( \dzT \dtau - \dztau \dT \right) = b
  \left( \int_0^1 \tau'^2 \, dz - \dztau \dtau \right) - \frac{b}{A} 
  \Qtiltb[\bv,\theta] . 
\end{equation}
where we have now defined the quadratic form
\begin{equation}
  \label{eq:Qtiltbdef}
  \begin{split}
    \Qtiltb[\bv,\theta] & = \Qtb[\bv,\theta] - \ltave{\intdf{\theta
        \unitn \cdot \del \theta}} \\
    & = \ltave{\intf \left[ \frac{1}{\Reff} |\del \bv|^2 + 2 \tau'
        w \theta + |\del \theta|^2 \right] - A \zlim{\have{\theta
          \theta_z}}} .
  \end{split}
\end{equation}
Equation \eqref{eq:bdid2}, which is still fully independent of
thermal BCs (subject to the symmetry condition $\tau'(0) = \tau'(1) =
-\dztau$), is the governing identity underlying our upper bounding
principle.

We comment that the prime in the notation $\Qtiltb$ refers to the
addition of the boundary terms to $\Qtb$ (no implied differentiation),
and note that for both fixed temperature and fixed flux BCs, the
boundary term $\tave{\zlim{\have{\theta \theta_z}}}$ vanishes, so
$\Qtiltb = \Qtb$.  In these cases, \eqref{eq:bdid2} thus reduces to
\begin{equation}
  \label{eq:bdiddir}
  \Nu - 1 = \dzT - 1 = b \left( \int_0^1 \tau'^2 \, dz - 1 \right) -
  \frac{b}{A} \Qtb [\bv,\theta] 
\end{equation}
for Dirichlet thermal BCs (for which $\dT = \dtau = 1$), and to 
\begin{equation}
  \label{eq:bdidneu}
  1 - \Nu^{-1} = 1 - \dT = b \left( \int_0^1 \tau'^2 \, dz - 2 \dtau + 1
  \right) - \frac{b}{A} \Qtb [\bv,\theta] .
\end{equation}
for Neumann thermal BCs, in which case $\dzT = \dztau = 1$.  For
general thermal BCs, $\dzT$ and $\dT$ are both \emph{a priori}
unknown, but they are related via the boundary conditions (see for
instance \eqref{eq:betadelta} below), so
that \eqref{eq:bdid2} may be written in terms only of either $\dzT$ or
$\dT$.

\subsection{Allowed fields, admissible backgrounds and the spectral
  constraint}
\label{ssec:admiss}

As formulated thus far, the general governing equation
\eqref{eq:bdid2} is an identity.  If, for given thermal BCs, one had
access to $\bv(\bx,t)$ and $\theta(\bx,t)$ satisfying
\eqref{eq:fluctmom}--\eqref{eq:fluctheat}, or sufficient information
about them to compute the time-averaged quadratic form $\Qtiltb
[\bv,\theta]$ \eqref{eq:Qtiltbdef}, then $\dzT$ and $\dT$, and hence
$\Nu$, could in principle be computed.  However, for turbulent
convection such analytical information is well beyond the limits of
what is (currently) accessible.  The fundamental insight underlying
upper bounding methods for convection is firstly, that if for some
$\tau$ and $b$, $\Qtiltb$ can be shown to be bounded below,
then this yields, ultimately, an \emph{upper bound} on the Nusselt
number $\Nu$ (for a given $\R$); and secondly, that such a lower bound
on $\Qtiltb$ may indeed often be demonstrated provided one is prepared to
widen the class of fields $\bv$, $\theta$ over which the minimization
takes place, as long as this class contains all solutions of
\eqref{eq:fluctmom}--\eqref{eq:fluctheat}.  The cost of weakening the
constraints on $\bv$ and $\theta$ is that this may reduce the lower
bound on $\Qtiltb$ and thereby weaken the upper bound estimate for
$\Nu$.

\subsubsection{Allowed fields $\bv$, $\theta$:}
\label{sssec:allowed}

In considering the class of allowable flows $\bv = \bu$ and
(fluctuation) temperature fields $\theta = T - \tau$ over which to
minimize $\Qtiltb$, we observe that if the dynamical constraints
\eqref{eq:fluctmom}--\eqref{eq:fluctheat} on the fields $[\bv,\theta]$
are removed without being replaced by assumptions on the temporal
structure or correlations of the fields---detailed knowledge of which
is unavailable---it becomes sufficient to minimize the quadratic form
$\Qtiltb$ over stationary fields $\bv(\bx)$ and $\theta(\bx)$.  The
conditions we can assume these fields $\bv(\bx)$ and $\theta(\bx)$ to
satisfy are: appropriate homogeneous boundary conditions, the
incompressibility constraint on the velocity fluctuations, and
boundedness of $\Ltf{\bv}$ and $\Ltf{\theta}$ (see
Appendix~\ref{app:bounded}).
More precisely, we denote the ``allowed'' fields $\bv$ and $\theta$
over which we minimize $\Qtiltb$ to be defined on
$[0,L_x]\times[0,L_y]\times[0,1]$, periodic in the horizontal
directions, so that $\del \cdot \bv = 0$, $\bv = \bO$ on $z = 0, 1$,
and $\theta$ satisfies homogeneous thermal boundary conditions
consistent with those of $T$.  For fixed temperature convection,
plausible but non-rigorous regularity assumptions restricting the
class of allowed fields have been shown to improve the scaling of the
$\Nu$-$\Ra$ bounds (\cite{CoDo96,Kers01}), but we shall not pursue
such assumptions here.

Observe that according to the above description, the trivial fields
$\bv = \bO$, $\theta = 0$ are ``allowed'', though they do not satisfy
\eqref{eq:fluctmom}--\eqref{eq:fluctheat} for $\tau'' \not= 0$.  Since
furthermore, $\Qtiltb[\lambda \bv,\lambda \theta] = \lambda^2
\Qtiltb[\bv,\theta]$ for $\lambda \in \reals$, it follows that if
$\Qtiltb$ is bounded below, then the minimum is zero.

\subsubsection{Admissible backgrounds:}
\label{sssec:admissible}

For each $\Reff > 0$ (that is, for each $\R > 0$ and $b>1$), we
thus denote a background field $\tau(z)$ \emph{admissible} if it
satisfies the appropriate BCs, the same as those for $T$ at the upper
and lower interfaces; and if the resultant quadratic form $\Qtiltb$ is
non-negative, $\Qtiltb[\bv,\theta] \geq 0$ for all allowed fields
$\bv$ and $\theta$.

Note that the positivity condition on the quadratic form
$\Qtiltb[\bv,\theta] \geq 0$ is equivalent to
\begin{equation}
  \label{eq:spec1}
  \Lambda(\tau,\Reff) = \mathop{\inf_{\bv,\theta}}_{\Ltf{\bv}^2 +
    \Ltf{\theta}^2 \not= 0} \left(
    \frac{\Qtiltb[\bv,\theta]}{\Ltf{\bv}^2 + \Ltf{\theta}^2} \right)
  \geq 0,
\end{equation}
where the infimum is taken over allowed fields $\bv$ and $\theta$.
Via the associated Euler-Lagrange equations, the condition
\eqref{eq:spec1} is equivalent to requiring that the linear operator
\begin{equation}
  \label{eq:specop}
  \ELop \begin{pmatrix} \bv \\ \theta \end{pmatrix} = 
  \begin{pmatrix}
    - \frac{1}{\Reff} \lap \bv + \tau' \theta + \del p \\
    - \lap \theta + \tau' w
  \end{pmatrix}
\end{equation}
acting on allowed fields $\bv$, $\theta$ has a positive semi-definite
spectrum, or that the lowest eigenvalue $\Lambda_0$ of $\ELop$ is
non-negative.  Consequently, the admissibility criterion on background
fields $\tau(z)$ (for a given $\Reff = b\R/(b-1) > 0$), equivalent to
\eqref{eq:spec1}, is also referred to as a \emph{spectral constraint}
on $\tau$ (\cite{DoCo96}).

\subsubsection{Fourier formulation of admissibility condition:}
\label{sssec:Fourieradmiss}

Due to the horizontal periodicity of the problem, we may reformulate
the admissibility condition $\Qtiltb \geq 0$ in horizontally
Fourier-transformed variables: we write the vertical component of
velocity $w$ and temperature fluctuation $\theta$ as
\begin{equation}
  \label{eq:wthk}
  w(x,y,z) = \sum_{\bk} \e^{\ii (k_x x + k_y y)} \whk(z) ,
  \qquad \theta(x,y,z) = \sum_{\bk} \e^{\ii (k_x x + k_y y)}
  \thk(z) ,  
\end{equation}
where the horizontal wave vector is $\bk = (k_x,k_y) = (2\pi
n_x/L_x,2\pi n_y/L_y)$, and we write $k^2 = |\bk|^2$; we shall also
write $\conj{\whk}$ for the complex conjugate of $\whk$ and $\D =
d/dz$.  We can use incompressibility to express the transformed
horizontal components of velocity in terms of the vertical component,
so that the admissibility criterion may be written completely in terms
of the Fourier modes $\whk$ and $\thk$.  This considerably simplifies
the formulation, particularly since $\Qtiltb$ is a quadratic form
(equivalently, the Euler-Lagrange equations for the minimization
problem are linear), so that different horizontal Fourier modes
decouple.  The no-slip boundary condition and incompressibility imply
that the BCs for $\whk(z)$ are $\whk = \D \whk = 0$ for $z = 0, 1$;
the BCs on $\thk$ obviously depend on the choice of thermal BCs, which
have so far been left unspecified.  We note also that $\wh_{\bO} = 0$;
this follows from incompressibility and horizontal periodicity via
$A \, \partial \have{w}/\partial z = \iint_A w_z \, dx \, dy = - \iint_A
(u_x + v_y) \, dx \, dy = 0$, which implies using $\have{w}|_{z = 0} =
0$ that $\have{w} = 0$ for all $z$.

Substituting \eqref{eq:wthk} into \eqref{eq:Qtbdef} and using
incompressibility, as in \cite{OWWD02} we can write the quadratic form
$\Qtb$ evaluated on allowed (stationary) fields $\bv$ and $\theta$ as
\begin{equation}
  \label{iq:QQk}
  \Qtb [\bv,\theta] =  \intf \left[ \frac{1}{\Reff} |\del
      \bv|^2 + 2 \tau' w \theta + |\del \theta|^2 \right] \geq A
    \sum_{\bk} \Qk,  
\end{equation}
where (see \cite{CoDo96,Kers01})
\begin{align}
  \Qk \equiv \Qktb [\whk, \thk] & = \int_0^1 \left[
    \frac{1}{\Reff} \left( k^2 |\whk|^2 + 2 |\D \whk|^2 +
      \frac{1}{k^2} |\D^2 \whk|^2 \right) + 2 \tau' \Real{\whk \conj{\thk}}
  \right. \nonumber \\
  \label{eq:Qkdef}
  & \qquad \quad \ \left. +
    \left( k^2 |\thk|^2 + |\D \thk|^2 \right) \right] \, dz \ ; 
\end{align}
note that \eqref{iq:QQk} is an equality for two-dimensional flows.
Since the boundary terms in \eqref{eq:Qtiltbdef} are expressed in
Fourier space as $\zlim{\have{\theta \theta_z}} = \sum_{\bk}
\left. \Real{\thk (z) \, \D \conj{\thk}(z)} \right|_{z=0}^1$, we thus
define
\begin{equation}
  \label{eq:Qktdef}
  \Qkt \equiv \Qkttb [\whk, \thk] = \Qk - \left. \Real{\thk (z) \,
    \D \conj{\thk}(z)} \right|_{z=0}^1 , 
\end{equation}
to give
\begin{equation}
  \label{iq:QtQkt}
  \Qtiltb [\bv,\theta] =  \intf \left[ \frac{1}{\Reff} |\del
      \bv|^2 + 2 \tau' w \theta + |\del \theta|^2 \right] - A
    \zlim{\have{\theta \theta_z}} \geq A \sum_{\bk} \Qkt . 
\end{equation}

Since the class of allowed fields $[\bv, \theta]$ includes fields
containing a single horizontal Fourier mode, it is now clear that
$\Qtiltb$ is a positive quadratic form, $\Qtiltb [\bv, \theta] \geq 0$
for all allowed fields $\bv$ and $\theta$, if and only if all the
quadratic forms $\Qkt = \Qkttb$ are positive.  Thus
the \emph{admissibility criterion} for background fields $\tau(z)$
(for given $\Reff > 0$) may be formulated in Fourier space,
as the condition that $\Qkt [\whk,\thk] \geq 0$ for all $\bk$ and for
all sufficiently smooth (complex-valued) functions $\whk(z)$,
$\thk(z)$ satisfying $\whk = \D \whk = 0$ at $z = 0,1$ and the
appropriate boundary conditions on $\thk$ at $z = 0,1$.

\subsubsection{Bounding principle:}
\label{sssec:genprinciple}

The thermal BCs at the plates imply an equation relating $\dzT$ and
$\dT$, whenever they do not specify either $\dT$ (fixed temperature) or
$\dzT$ (fixed flux) directly.  Thus, for instance (for BCs other than
fixed flux), we can substitute for $\dT$ and write \eqref{eq:bdid2} in
terms of only $\dzT$, as done below (the
fixed temperature and fixed flux cases are given in
\eqref{eq:bdiddir}--\eqref{eq:bdidneu}).

In such a case for a given $\R$ and $b$, for any admissible background
field $\tau(z)$ the inequality $\Qtiltb \geq 0$ implies an upper bound
$\BdzT \taubfunc$ on the averaged boundary flux $\dzT$.  Via the
relation between $\dzT$ and $\dT$, this also gives a lower bound $\BdT
\taubfunc$ on the averaged temperature drop across the fluid; using
\eqref{eq:Nueq}, in this way admissible backgrounds lead to upper
bounds $\Nubnd \taubfunc$ for $\Nu$.  For a given $\R > 0$, the best
upper bound $\NuboundR$ on $\Nu$ obtainable using this approach is now
obtained by minimizing $\Nubnd \taubfunc$ over admissible $\tau(z)$
and $b$.\footnote{Recall that the class of admissible $\tau(z)$
  depends on $\R$ through $\Reff$.}  Finally, the relationship $\Ra =
\R \, \dT \geq \R \, \BdT \taubfunc$ lets us bound $\R$, and hence
$\Nu$, from above as a function of $\Ra$.

\section{Mixed (Robin) thermal boundary conditions}
\label{sec:mixedth}

The formulation and derivations above were developed independent of
thermal boundary conditions, except that we have restricted our
attention to fluids with thermally identical upper and lower
boundaries, which permits the symmetry assumption $\tau'(0) = \tau'(1)
= - \dztau$.  We now specialize to particular BCs to make further
progress: General linear conditions at the boundary of a fluid as in
figure~\ref{fig:RBsetup} are of mixed (Robin) type.  In a subsequent
paper we consider the more realistic case of a fluid in thermal
contact with bounding plates of finite thickness and conductivity.

\subsection{Fixed Biot number boundary conditions and
  nondimensionalization}
\label{ssec:biotnondim}

In dimensional terms, we choose the mixed (Robin) BCs on the plates to
take the form
\begin{equation}
  \label{eq:robbcdim}
  T^* + \Bi^* \unitn \cdot \del^* T^* = \Tlb \ \ \text{on}\ \ z^* = 0, \qquad
  T^* + \Bi^* \unitn \cdot \del^* T^* = \Tub \ \ \text{on}\ \ z^* = h.
\end{equation}
for some given constant $0 \leq \Bi^* < \infty$.\footnote{The limit
  $\Bi^* \to \infty$ is treated by writing \eqref{eq:robbcdim} in the
  equivalent form (for $\Bi^* > 0$) $\unitn \cdot \del^* T^* +
  T^*/\Bi^* = B^*_{l,u}$ on $z^* = 0,h$, where (for $0 < \Bi^* <
  \infty$) $B^*_{l,u} = A^*_{l,u}/\Bi^*$.}  These conditions may be
interpreted as Newton's Law of Cooling (Heating), in which the
boundary heat flux is assumed proportional to the temperature change
across the boundary: $- \kcof \unitn \cdot \del^* T^* = \kcof (T^* -
\Tlb)/\Bi^*$.

We use $\unitn = -\unitz, +\unitz$ on $z^* = 0, h$ respectively, and
nondimensionalize by substituting $z^* = h z$, $T^* = \Tr + \Tscal
T$.  Defining the \emph{Biot number} $\Bi = \Bi^*/h$, we
find\footnote{There appears to be little consensus in the literature
  as to whether the term ``Biot number'' refers to $\Bi$ as defined in
  \eqref{eq:nondimcond}, or to its inverse $\Bi^{-1}$.}
\begin{equation}
  \label{eq:nondimcond}
  T - \Bi \, T_z = \frac{\Tlb - \Tr}{\Tscal}\ \ \text{on}\ \ z = 0,
  \qquad
  T + \Bi \, T_z = \frac{\Tub - \Tr}{\Tscal}\ \ \text{on}\ \ z = 1.  
\end{equation}
At the moment $\Tr$ and $\Tscal$ are still unspecified.  A convenient
choice, consistent with the nondimensionalizations introduced previously for
the limiting fixed temperature and fixed flux cases, is to require the
conducting state ($\bu^* = \bO$, $\del^* T^* = C \, \unitz$ for some
constant $C < 0$) to take the form $\bu = \bO$, $T = 1 - z$ in the
nondimensional variables.  The condition that $T = 1-z$ satisfies the
nondimensional BCs \eqref{eq:nondimcond} implies $(\Tlb - \Tr)/\Tscal
= 1 + \Bi$, $(\Tub - \Tr)/\Tscal = -\Bi$, so that (for $\Bi < \infty$)
our chosen temperature scale and reference temperature are
\begin{equation}
  \label{eq:nondimchoice}
  \Tscal = \frac{\Tlb - \Tub}{1 + 2\Bi}, \qquad \Tr = \frac{\Tub + \Bi
  (\Tlb + \Tub)}{1 + 2\Bi}. 
\end{equation}
Having finally fixed a choice of dimensionless variables, the
nondimensional mixed thermal boundary conditions (fixed Biot number)
are
\begin{equation}
  \label{eq:robbc}
  T - \Bi \, T_z = 1 + \Bi \ \ \text{on}\ \ z = 0, \qquad
  T + \Bi \, T_z = - \Bi \ \ \text{on}\ \ z = 1.
\end{equation}

Note that the mixed (Robin) BCs \eqref{eq:robbc} reduce to the fixed
temperature (Dirichlet) BCs \eqref{eq:dirbc} in the limit $\Bi \to 0$,
and to the fixed flux (Neumann) BCs \eqref{eq:neubc} in the limit $\Bi
\to \infty$; thus we denote $\Bi = 0$ and $\Bi = \infty$ as the
``fixed temperature'' and ``fixed flux'' cases, respectively.

\subsection{Governing identities for fixed finite Biot number}
\label{ssec:mixedid}

In the general case, neither the boundary temperature drop $\dT$ nor
the flux $\dzT$ that combine in the computation \eqref{eq:Nueq} of the
Nusselt number is known \emph{a priori}.  However, we can derive a
relation between them: taking horizontal averages of \eqref{eq:robbc},
and subtracting the upper boundary condition from the lower, we find
$\have{T}|_{z=0} - \have{T}|_{z=1} - \Bi(\have{T}_z|_{z=0} +
\have{T}_z|_{z=1}) = 1 + 2\Bi$. 
Taking time averages and using \eqref{eq:deltadef} and
\eqref{eq:betadef}, we find the fundamental relation
\begin{equation}
  \label{eq:betadelta}
  \dT + 2 \Bi \dzT = 1 + 2\Bi 
\end{equation}
(this formula also holds in the fixed temperature and flux limits $\Bi
\rtarr 0$ and $\Bi \rtarr \infty$).  
Hence for $0 < \Bi < \infty$, an upper bound on $\dzT$ constitutes a
lower bound on $\dT$, and \emph{vice versa}, and we only need to bound one of
these quantities to obtain an upper bound on $\Nu = \dzT/\dT = \dzT/[1
+ 2\Bi (1 - \dzT)] = \dT^{-1} + \left( \dT^{-1} - 1 \right)/2\Bi$.

Using the identity \eqref{eq:betadelta} to solve for either $\dT = 1 +
2\Bi (1 - \dzT)$ or $\dzT = 1 + (1 - \dT)/2\Bi$ and substituting into
the results of
Sections~\ref{ssec:globalenergy}--\ref{ssec:genformulation}, we obtain
the forms of the governing energy identities for mixed thermal BCs.
We shall state these identities in a way that permits us to obtain an
upper bound on $\dzT$ (the relations are stated in terms of $\dzT$,
$\dztau$, $\tave{\have{T_z^2}|_{z=0,1}}$ and
$\tave{\have{\theta_z^2}|_{z=0,1}}$, valid for $\Bi < \infty$); this
is the formulation suitable for small $\Bi$, which reduces to the
corresponding previously stated identities for fixed temperature BCs
in the limit $\Bi \to 0$.  For completeness, in
Appendix~\ref{app:ffidentities} we give the forms of the identities
(equivalent for $0 < \Bi < \infty$) which yield the fixed flux limit
$\Bi \to \infty$.

First, solving for $\dT$ from \eqref{eq:betadelta} and substituting
into the global kinetic energy identity \eqref{eq:gradu2}, we find
\begin{equation}
  \frac{1}{A R} \tave{\Ltf{\del \bu}^2 } = \dzT - \dT = ( 1 +
    2\Bi) \left( \dzT - 
    1 \right)  \label{eq:gradurobd}
\end{equation}
which reduces to \eqref{eq:gradudir} in the limit $\Bi \rtarr 0$.
Similarly, we can evaluate the boundary term in \eqref{eq:gradT} by
using the BCs \eqref{eq:robbc} to solve for $T$ at $z = 0,1$ in terms
of $T_z$; substituting into $\zlim{ T T_z}$ and taking horizontal and
time averages,
we find that the global thermal energy identity \eqref{eq:gradT}
becomes
\begin{equation}
  \frac{1}{A} \tave{\Ltf{\del T}^2 } = \tave{\zlim{\have{T T_z}}}
  = (1 + 2\Bi ) \dzT - \Bi 
  \tave{ \have{T_z^2}|_{z=0} + \have{T_z^2}|_{z=1}} \label{eq:gradTrobd}
\end{equation}
which again reduces to the appropriate fixed temperature limit \eqref{eq:gradTdir}.

In the background flow formulation, the requirement for the background
field to satisfy the given BCs in this case means that $\tau(z)$
should obey \eqref{eq:robbc}, implying that
\begin{equation}
  \label{eq:gammadtau}
  \dtau + 2 \Bi \dztau = 1 + 2 \Bi 
\end{equation}
using \eqref{eq:gammadtaudef}.
Consequently, the perturbation $\theta$ satisfies the homogeneous Robin
BCs $\theta + \Bi \unitn \cdot \del \theta = 0$ at the interfaces,
which in our geometry become 
\begin{equation}
  \label{eq:throbbc}
  \theta - \Bi\, \theta_z = 0 \ \ \text{at}\ \ z = 0, \qquad 
  \theta + \Bi\, \theta_z = 0 \ \ \text{at}\ \ z = 1,
\end{equation}
and translate for horizontal Fourier modes defined in
\eqref{eq:wthk} to the BCs
\begin{equation}
  \label{eq:thkrobbc}
  \thk(0) = \Bi \, \D \thk(0) , \qquad \thk(1) = - \Bi \, \D \thk(1) .
\end{equation}
The boundary term in \eqref{eq:Qtiltbdef} is thus $\intdf{\theta \,
  \unitn \cdot \del \theta} = - \Bi \intdf{(\unitn \cdot \del
  \theta)^2} \leq 0$, or
\begin{equation}
  \label{eq:thbdyrobd}
  \zlim{\have{\theta \theta_z}} = - \Bi (\have{\theta_z^2}|_{z=0} + 
  \have{\theta_z^2}|_{z=1})  ,
\end{equation}
which can equivalently be written in Fourier space (see
\eqref{eq:Qktdef}) as
\begin{equation}
  \label{eq:thkbdyrobd}
  \zlim{\have{\theta \theta_z}} = \sum_{\bk} \zlim{\Real{\thk(z) \D
      \conj{\thk}(z)}} = -\Bi \sum_{\bk} \left( |\D \thk(0)|^2 + |\D
    \thk(1)|^2 \right) . 
\end{equation}
It follows that  $\Qtiltb[\bv,\theta] \geq \Qtb[\bv,\theta]$, so
that a lower bound on $\Qtb$ implies a lower bound on $\Qtiltb$; the
additional boundary term which appears for $\Bi \not= 0, \infty$ is
stabilizing. 

Finally, the form of the governing identity \eqref{eq:bdid2} for mixed
thermal BCs, fundamental to the formulation of a bounding principle,
may now be derived: Solving and substituting for $\dT$ and $\dtau$
using \eqref{eq:betadelta} and \eqref{eq:gammadtau}, we find
\begin{equation}
  \dzT \dtau - \dztau \dT + \dztau \dtau = (1 + 2\Bi )\dzT - 2 \Bi
  \dztau^2 . \label{eq:bgidd} 
\end{equation}
Now substituting \eqref{eq:gradurobd} and \eqref{eq:bgidd},
\eqref{eq:bdid2} becomes in terms of $\dzT$ (compare
\eqref{eq:bdiddir})
\begin{equation}
  \label{eq:bdidrobd}
  (1 + 2 \Bi) (\dzT - 1) = b \left( \int_0^1 \tau'^2 \, dz - (1 + 2
    \Bi) + 2 \Bi \dztau^2 \right) - \frac{b}{A} \Qtiltb[\bv,\theta] ,
\end{equation}
where we evaluate the boundary term in $\Qtiltb[\bv,\theta]$ using
\eqref{eq:thbdyrobd}.

\subsection{A bounding principle for mixed thermal boundary conditions}
\label{ssec:mixbound}

Following the general approach outlined in Section~\ref{ssec:admiss},
the governing equations derived in Section~\ref{ssec:mixedid} now
allow us to formulate an upper bounding principle for the Nusselt
number $\Nu$ in terms of the control parameter $\R$ (and hence in
terms of $\Ra$, via $\Ra = \R \, \dT$):

For each $\R > 0$, if we can choose $b > 1$ and a corresponding
admissible background field $\tau(z)$ (so that $\Qtiltb \geq 0$), then
from \eqref{eq:bdidrobd} the averaged boundary temperature gradient
$\dzT$ is bounded above by
\begin{equation}
  \label{iq:betabnd}
  \dzT \leq 1 - b + \frac{b}{1 + 2\Bi} \left( \int_0^1 \tau'^2\,dz +
  2\Bi \dztau^2 \right) = \BdzT_{\Bi} \taubfunc ,
\end{equation}
while from \eqref{eq:bdidrobn}, the averaged temperature drop across
the fluid $\dT$ is bounded below by 
\begin{equation}
  \label{iq:deltabnd}
  \dT \geq 1 + b (2\dtau - 1) - b \, \frac{2\Bi}{1 + 2\Bi} \left(
    \int_0^1 \tau'^2\,dz + \frac{1}{2\Bi} \dtau^2 \right) = \BdT_{\Bi}
  \taubfunc , 
\end{equation}
where the above equations define the functionals $\BdzT_{\Bi}
\taubfunc$ and $\BdT_{\Bi} \taubfunc$.  (Of course, for $0 < \Bi <
\infty$ the bounds \eqref{iq:betabnd} and \eqref{iq:deltabnd} are not
independent; in principle we only need to find, say, the upper bound
$\BdT_{\Bi} \taubfunc$ for $\dzT$ (for $\Bi < \infty$), since by
\eqref{eq:betadelta} we have $\BdT_{\Bi} \taubfunc + 2\Bi \BdzT_{\Bi}
\taubfunc = 1 + 2\Bi$.)  An upper bound on $\dzT$ and a corresponding
lower bound on $\dT$ then imply via \eqref{eq:Nueq} that the Nusselt
number is bounded above by
\begin{equation}
  \label{iq:Nubnd}
  \Nu \leq \Nubnd_{\Bi} \taubfunc = \BdzT_{\Bi} \taubfunc/\BdT_{\Bi}
  \taubfunc .
\end{equation}

Observe that for the conduction solution $\tau(z) = 1-z$, we have
$\BdzT_{\Bi} \taubfunc = \BdT_{\Bi} \taubfunc = 1$, so that whenever this is an
admissible profile, the bound on the Nusselt number takes its minimum
value of 1, as expected.

\section{Piecewise linear background and elementary estimates}
\label{sec:pwlinest}

As discussed in Section~\ref{ssec:admiss}, the best upper bound
$\NuboundR$ that may be achieved by the above formulation for each
value of $\R$ is obtained by optimizing the upper bounds on the
Nusselt number $\Nubnd \taubfunc$ over all admissible $\tau(z)$ and
over $b > 1$.  Careful numerical studies obtaining such optimal
solutions of analogous bounding problems have been performed for plane
Couette flow (which may be related to fixed temperature convection) by
\cite{PlKe03} and for infinite Prandtl number convection by
\cite{IKP06}.  However, by restricting the class of admissible
backgrounds $\tau(z)$ over which the optimization is performed, upper
bounds may be obtained much more readily, at the (likely) cost of
weakening the upper bound.

\subsection{Piecewise linear background profiles}
\label{ssec:pwlinear}

Following \cite{DoCo96} and subsequent works, we thus introduce a
one-parameter family of piecewise linear background profiles
$\taud(z)$, for which $\taud' = -\dztau$ for $0 \leq z < \delta$ and
$1-\delta < z \leq 1$:
\begin{equation}
  \label{eq:taudelta}
  \tau(z) = \taud(z) = \left \{
    \begin{array}{ll}
      \avetau - \dztau(z-\delta) , & 0 \leq z \leq \delta , \\
      \avetau , & \delta < z < 1 - \delta , \\
      \avetau - \dztau(z - (1-\delta)) , & 1-\delta \leq z \leq 1 ; \\
    \end{array}
  \right.
\end{equation}
see figure~\ref{fig:pwlin_biot}.
\begin{figure}
  \begin{center}
        \includegraphics[width = 3.1in]{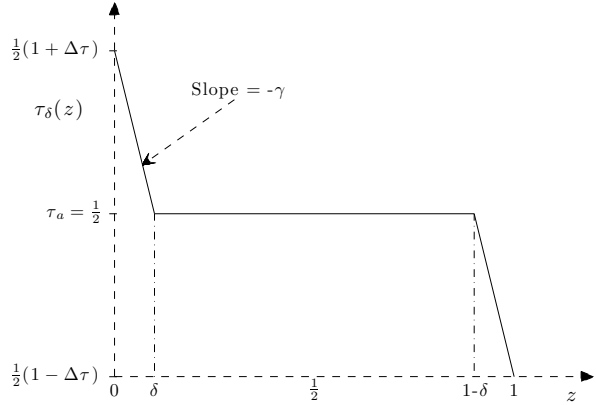}
    \caption{The piecewise linear background profile
      $\tau_{\delta}(z)$, with $\tau' = -\dztau$ in the boundary
      layer, and $\tau' = 0$ in the bulk.}
    \label{fig:pwlin_biot}
  \end{center}
\end{figure}%
Here the single parameter $\delta \leq 1/2$ may be interpreted as
modelling the thickness of the thermal boundary layer.
The intuition behind this definition is that in order for $\tau(z)$ to
be admissible, the indefinite term $\intf 2 \tau' w \theta$ in
$\Qtiltb [\bv, \theta]$ (see \eqref{eq:Qtbdef}, \eqref{eq:Qtiltbdef})
should be controlled by the other, positive terms.  With this choice
of background, $2 \tau' w \theta$ vanishes in the bulk of the domain,
and is nonzero only near the fluid boundaries, where $w$ is small.
Furthermore, since $\tau'$ is piecewise constant, explicit analytical
bounds are readily attainable, giving (non-optimal) rigorous bounds on
the Nusselt number.

From the definition \eqref{eq:taudelta}, we immediately compute
$\tau(0) = \avetau + \dztau \delta$, $\tau(1) = \avetau - \dztau
\delta$, and so 
\begin{equation}
  \label{eq:pwlvals}
  \dtau = \tau(0) - \tau(1) = 2 \delta \dztau, \qquad
  \int_0^1 \tau'^2 \, dz = 2 \delta \dztau^2 = \dztau \dtau ,
\end{equation}
where it remains to choose the average $\avetau = \hlf (\tau(0)
+ \tau(1))$ and boundary slope $\dztau = -\tau'(0) = -\tau'(1)$ of the
background as functions of $\delta$.
For the mixed (Robin) thermal BCs introduced in
Section~\ref{sec:mixedth}, substituting $\dtau = 2 \delta \dztau$
into the relation $\dtau + 2 \Bi \dztau = 1 + 2\Bi$
\eqref{eq:gammadtau} we obtain the value of $\dztau$ (for given
$\delta$ and $\Bi$), as well as the corresponding $\dtau$,
for which the piecewise linear profile \eqref{eq:taudelta} satisfies
the BCs:
\begin{equation}
  \label{eq:gammaBi}
  \dztau = \frac{1 + 2\Bi}{2(\delta + \Bi)} , \qquad
  \dtau = 2\delta \dztau = \frac{\delta (1+2\Bi)}{\delta+\Bi} ,
\end{equation}
which implies the related identity $1-\dtau = 2\Bi (\dztau - 1)$;
observe that since $\delta \leq 1/2$, we have $1 \leq \dztau \leq
1/2\delta$ and $\dtau \leq 1$.  Using this $\dztau$ in the BC
\eqref{eq:robbc} $\tau(0) = 1 + \Bi + \Bi \tau'(0)$ for $\Bi <
\infty$, we find that $\avetau = 1/2$ (so that we may write $\tau(0) =
\hlf (1 + \dtau)$, $\tau(1) = \hlf (1 - \dtau)$), completing the
specification of the background $\taud(z)$.

Substituting formulas \eqref{eq:pwlvals}--\eqref{eq:gammaBi} into
\eqref{iq:betabnd}--\eqref{iq:deltabnd} and simplifying, we now find
that the conservative bounds on $\dzT$ and $\dT$ for fixed Biot number
convection with a piecewise linear (pwl) background profile $\taud$
take the concise forms
\begin{align}
  \label{eq:dzTrobpwlin}
  \dzT \leq \BpwldzT{\Bi} (\delta,b) \equiv \BdzT_{\Bi} \taudbfunc & = 1 + b \,
  \frac{1}{2} \, \frac{1 - 
    2\delta}{\delta + \Bi} = 1 + b (\dztau - 1) , \\
  \label{eq:dTrobpwlin}
  \dT \geq  \BpwldT{\Bi} (\delta,b) \equiv \BdT_{\Bi} [\taud; b] & = 1 - b
  \, \Bi \, \frac{1 - 2\delta}{\delta + \Bi} = 1 + b (\dtau - 1) ,
\end{align}
and the corresponding upper bound on the Nusselt number is $\Nu \leq
\Nupwl{\Bi} (\delta,b) \equiv \Nubnd_{\Bi} \taudbfunc =
\BpwldzT{\Bi}(\delta,b)/\BpwldT{\Bi}(\delta,b)$.  Since $b > 0$,
these bounds satisfy $\BpwldzT{\Bi}(\delta,b) \geq 1$,
$\BpwldT{\Bi}(\delta,b) \leq 1$, and hence $\Nupwl{\Bi}(\delta,b) \geq
1$, as one might expect. 
Observe that the bounds $\BpwldzT{\Bi}(\delta,b)$ and
$\BpwldT{\Bi}(\delta,b)$ do not depend explicitly on the control
parameter $\R$, but rather indirectly through the admissibility
condition on $\delta$.

In the special case of fixed temperature BCs, for which $\dT = \dtau
=1$, we have $\dztau = 1/2\delta = \int_0^1 \tau'^2 \, dz$, and the
bound on $\dzT$, and hence on the Nusselt number, becomes
\begin{equation}
  \label{eq:dzTdirpwlin}
  \Nu = \dzT \leq \BpwldzT{0}(\delta,b) = 1 + b \left(
    \frac{1}{2\delta} - 1 \right) . 
\end{equation}
At the opposite extreme, the fixed flux BCs impose $\dzT = \dztau =
1$, so that we must choose $\dtau = 2\delta = \int_0^1 \tau'^2 \, dz$;
then the lower bound on $\dT$ (corresponding to an upper bound on
$\Nu$) is
\begin{equation}
  \label{eq:dTneupwlin}
  \Nu^{-1} = \dT \geq \BpwldT{\infty}(\delta,b) = 1 + b (2\delta - 1) .
\end{equation}
Again, in this formulation the fixed temperature and fixed flux cases
are the $\Bi \to 0$ and $\Bi \to \infty$ limits of the bounds for
general Biot number.  We note, however, that the thermal BCs of the form
\eqref{eq:robbc} do not specify the value of $\avetau$ in the fixed
flux case $\Bi = \infty$ (the governing equations depend only on
temperature gradients, not on their absolute values); here we choose
$\avetau = 1/2$ for convenience, but use a different choice in the
boundedness proof of Appendix~\ref{app:bounded}.

\subsection{Cauchy-Schwarz estimates on the quadratic form}
\label{ssec:CSestimates}

Recall the admissibility criterion for the background field
$\taud(z)$: $\Qcalt_{\taud,\Reff} [\bv,\theta] \geq 0$ for all allowed
$\bv$ and $\theta$, or in Fourier space (by
\eqref{iq:QQk}--\eqref{iq:QtQkt}) $\Qkt = \Qcalt_{\bk; \taud, \Reff}
[\whk,\thk] \geq 0$ for all $\bk$. For piecewise linear background
fields $\taud(z)$ of the form \eqref{eq:taudelta}, this criterion
reduces to a requirement that $\delta$ is sufficiently small, for
given $\Reff = b \R/(b-1)$.  Elementary Cauchy-Schwarz and Young
inequalities applied to the Fourier space quadratic form $\Qkt$ allow
us to derive explicit \emph{sufficient} conditions on $\delta$ so that
$\Qkt \geq 0$ for all $\bk$, and hence to estimate upper bounds on
$\Nu$.

We recall first that the boundary terms in $\Qcalt_{\taud,\Reff}
[\bv,\theta]$ for $\Bi \in (0,\infty)$ are nonnegative for Robin BCs
(and vanish if $\Bi = 0$ or $\Bi = \infty$), so that it is enough
to verify the admissibility criterion for $\Qcal_{\taud,\Reff}
[\bv,\theta]$ (see the discussion below \eqref{eq:thkrobbc}).
Equivalently in Fourier space, using \eqref{eq:thkbdyrobd} to evaluate
the boundary term in \eqref{eq:Qktdef} for $\Bi \in (0,\infty)$, we
have
\begin{equation}
  \label{iq:QktQkrob}
  \Qkt = \Qk + \Bi  \left( |\D \thk(0)|^2 + |\D \thk(1)|^2
  \right) \geq \Qk , 
\end{equation}
so that it suffices to obtain conditions on $\delta$ to ensure $\Qk
\geq 0$ for all $\bk$.
To do so, we need to control the only indefinite term in $\Qk$,
$\int_0^1 2\taud' \Real{\whk \conj{\thk}}$, by the other terms.  For
completeness we review the necessary estimates from \cite{OWWD02}:
Since $\whk$ and $\D \whk$ (and hence also $\whk \, \conj{\thk}$)
vanish at both boundaries, we have
\begin{equation}
  \label{iq:whkthkz}
  |\whk(z) \, \conj{\thk}(z)| = \left| \int_0^z \D \left( \whk
      \conj{\thk} \right) \, d\zeta \right| \leq \int_0^z | \whk
    \D \conj{\thk} | \, d\zeta + \int_0^z | \conj{\thk} \D
    \whk | \, d\zeta ,
\end{equation}
where for $0 \leq z \leq \hlf$, by the Fundamental Theorem of Calculus
and the Cauchy-Schwarz inequality we find that
\begin{align}
  \label{iq:whkCS}
  \left| \whk(z) \right| = \left| \int_0^z \D \whk \, d\zeta \right|
  & \leq \sqrt{z} \left( \int_0^z \left| \D \whk (\zeta) \right|^2 \,
    d\zeta\right)^{1/2} \leq \sqrt{z} \| \D \whk \|_{[0,\hlf]} , \\
  \label{iq:DwhkCS}
  \left| \D \whk(z) \right| = \left| \int_0^z \D^2 \whk \, d\zeta \right|
  & \leq \sqrt{z} \left( \int_0^z \left| \D^2 \whk (\zeta) \right|^2 \,
    d\zeta\right)^{1/2} \leq \sqrt{z} \| \D^2 \whk \|_{[0,\hlf]} . 
\end{align}

Substituting these estimates into \eqref{iq:whkthkz} and again
applying the Cauchy-Schwarz inequality, for $0 \leq z \leq \hlf$ we obtain 
\begin{align}
  |\whk(z) \, \conj{\thk}(z)| & \leq \left( \int_0^z \zeta \, d\zeta
  \right)^{1/2} \left[ \| \D \whk \|_{[0,\hlf]} \left( \int_0^z |\D
      \thk|^2 \, d\zeta \right)^{1/2} \right. \nonumber \\
  & \qquad \qquad \qquad \quad \ 
  \left. + \| \D^2 \whk \|_{[0,\hlf]}
    \left( \int_0^z |\thk|^2 \, d\zeta \right)^{1/2} \right] 
  \nonumber \\
  \label{iq:whkthkz2}
  & \leq \frac{z}{2\sqrt{2}} \left[ a_1 \| \D \whk
    \|_{[0,\hlf]}^2 + \frac{1}{a_1} \| \D \thk \|_{[0,\hlf]}^2 + 
    \frac{a_2}{k^2} \| \D^2 \whk \|_{[0,\hlf]}^2 + \frac{k^2}{a_2} \| \thk
    \|_{[0,\hlf]}^2 \right]  ,
\end{align}
where we have also applied Young's inequality $p q \leq \hlf (a_j p^2
+ q^2/a_j)$ for any $a_j > 0$.  Proceeding similarly, we obtain an
analogous estimate for $\hlf \leq z \leq 1$.  
For the piecewise linear background $\taud(z)$, for which $\tau' =
-\dztau < 0$ for $0 \leq z \leq \delta$ and $1-\delta \leq 1$, and $\tau'
= 0$ otherwise, applying these estimates we have
\begin{align*}
  \left| \int_0^1 \taud' \, \whk \conj{\thk} \, dz \right| & \leq
  \dztau \left( \int_0^{\delta} |\whk \conj{\thk}| \, dz +
    \int_{1-\delta}^1 |\whk \conj{\thk}| \, dz \right) \\
  & \leq \frac{\dztau \delta^2}{4 \sqrt{2}} \left[
    a_1 \| \D \whk \|_{[0,1]}^2 + \frac{1}{a_1} \| \D
    \thk \|_{[0,1]}^2 + \frac{a_2}{k^2} \| \D^2 \whk \|_{[0,1]}^2 +
    \frac{k^2}{a_2} \| \thk \|_{[0,1]}^2 \right] ,
\end{align*}
and thus
\begin{align}
  \int_0^1 2 \taud' \Real{\whk \conj{\thk}} \, dz & = \int_0^1 \taud'
  \left( \whk \conj{\thk} + \conj{\whk} \thk \right) \, dz
  \nonumber \\
  \label{iq:indef}
  & \geq - \frac{\dztau \delta^2}{2 \sqrt{2}} \left[
    a_1 \| \D \whk \|^2 + \frac{1}{a_1} \| \D
    \thk \|^2 + \frac{a_2}{k^2} \| \D^2 \whk \|^2 +
    \frac{k^2}{a_2} \| \thk \|^2 \right] ,
\end{align}
where norms are taken over the entire interval $[0,1]$ unless
otherwise indicated.  Substituting this estimate on the indefinite
term into $\Qk$ given by \eqref{eq:Qkdef}, we find
\begin{align*}
  \Qk & \geq \left( \frac{2}{\Reff} - \frac{\dztau
      \delta^2\, a_1}{2\sqrt{2}} \right) \| \D \whk\|^2 + \left(
    \frac{1}{\Reff} - \frac{\dztau \delta^2\, a_2}{2\sqrt{2}} \right)
  \frac{1}{k^2} \| \D^2 \whk \|^2 + \frac{1}{\Reff} k^2 \|\whk\|^2 \\
  & \qquad + \left( 1 - \frac{\dztau \delta^2}{2\sqrt{2}\, a_2}\right) k^2 \|
  \thk \|^2 + \left( 1 - \frac{\dztau \delta^2}{2\sqrt{2} \, a_1}
  \right) \| \D \thk \|^2 .
\end{align*}
In the absence of any additional \textit{a priori} information, for
instance on the decay rate of the Fourier coefficients (compare
\cite{CoDo96,Kers01}), our remaining estimates are necessarily
$k$-independent; we ensure the positivity of $\Qk$ by requiring all
coefficients to be nonnegative.  We choose $a_1 = a_2 = \dztau
\delta^2/2\sqrt{2}$; then, using \eqref{iq:QktQkrob} and dropping
manifestly nonnegative terms,
\begin{equation}
  \label{iq:Qkdconstr}
  \Qkt \geq \Qk \geq \left( \frac{2}{\Reff} - \frac{\dztau^2 \delta^4}{8}
  \right) \| \D \whk\|^2 + \left( \frac{1}{\Reff} - \frac{\dztau^2
      \delta^4}{8} \right) \frac{1}{k^2} \| \D^2 \whk \|^2 .
\end{equation}
We can thus guarantee that $\Qk \geq 0$ (and hence $\Qkt \geq 0$) if
we choose $\dztau^2 \delta^4/8 \leq 1/\Reff$.  
For given thermal BCs, $\dztau = \dztau(\delta)$ is specified as a
function of $\delta$; so this is a constraint on $\delta$ to have $\Qk
\geq 0$, that is, for $\taud(z)$ to be an admissible background.
Defining $\delta_c$ by
\begin{equation}
  \label{eq:deltaconstr}
  \dztau(\delta_c)^2 \delta_c^4 = \frac{8}{\Reff} = 8 \frac{b-1}{b \R} ,
\end{equation}
we obtain the best bound in this approach by choosing $\delta =
\delta_c$; the piecewise linear profile $\taud$ is admissible for any
$\delta \leq \delta_c$.

We observe that the estimates
\eqref{iq:whkthkz}--\eqref{iq:Qkdconstr}, and hence the sufficient
condition \eqref{eq:deltaconstr} on $\delta$, are independent of the
choice of boundary conditions on $\thk$ at $z = 0, 1$ (apart from
symmetry).  However, the thermal BCs enter the admissibility condition
on $\taud$ through the value \eqref{eq:gammaBi} of $\dztau =
\dztau(\delta)$ as a function of $\Bi$.

\section{Explicit asymptotic bounds for general thermal boundary
  conditions}
\label{sec:Bibound}

Using the piecewise linear background profile $\taud(z)$ and estimates
introduced in Section~\ref{sec:pwlinest}, we may now derive explicit
analytical bounds on the growth of the Nusselt number $\Nu$ with the
control parameter $\R$, and hence with the Rayleigh number $\Ra$, for
thermal boundary conditions with fixed Biot number $0 \leq \Bi \leq
\infty$.  We begin as usual by recalling the results for Dirichlet
($\Bi = 0$) and Neumann ($\Bi = \infty$) BCs, as in the general case
it then becomes apparent that the fixed temperature case is a singular
limit, while for \emph{any} $\Bi > 0$, the $\R \to \infty$ asymptotic
scaling is as in the fixed flux case.  In this Section we summarize
the main asymptotic bounds; more details including discussion of
different scaling regimes for $0 < \Bi < \infty$ will be given
elsewhere (\cite{WiGa08prep}; see also \cite{Gao06m}), while rigorous,
though somewhat weaker, bounds are proved in
Appendix~\ref{app:rigbounds}.

\subsection{Fixed temperature boundary conditions}
\label{ssec:dirbnd}

In the case of Dirichlet BCs, we have $\dT = \dtau = 1$, $\R = \Ra$,
and \eqref{eq:gammaBi} implies $\dztau = 1/2\delta$.  Thus the
sufficient condition \eqref{eq:deltaconstr} on $\delta$ simplifies to
$\delta \leq \delta_c$ where
\begin{equation}
  \label{eq:deltacdir}
  \delta_c^2 = \frac{32}{\Reff} = 32 \frac{b-1}{b \R} .
\end{equation}
One can show that the optimal choice of $b$ in this formulation is
$b_0 = 3/2$ (see \cite{WiGa08prep}), for which $\Reff = 3\R$, and hence
$\delta \leq \delta_c = 4 \sqrt{2/3} \, \R^{-1/2}$ is sufficient to
obtain a rigorous bound.  Since for this $b = b_0$,
\eqref{eq:dzTdirpwlin} becomes
\begin{equation}
  \label{eq:Nudpwldir}
  \Nu = \dzT \leq \BpwldzT{0} (\delta,b_0) = 1 - b_0 +
  \frac{b_0}{2\delta} = -\frac{1}{2} + \frac{3}{4\delta} ,
\end{equation}
for any $\delta \leq \delta_c$, the best rigorous analytical bound on
the Nusselt number using this approach is
\begin{equation}
  \label{eq:Nubnddir}
  \Nu \leq \tilde{\Nubnd}_{\text{pwl}} (\R) = \BpwldzT{0}
  (\delta_c,b_0) =   -\frac{1}{2} + \frac{3}{4\delta_c} =
  -\frac{1}{2} + \frac{3}{16} \sqrt{\frac{3}{2}} \R^{1/2} =
  - \frac{1}{2} + \frac{3 \sqrt{6}}{32} \Ra^{1/2} ,
\end{equation}
where we used the fact that for fixed temperature BCs, the control
parameter $\R$ is the usual Rayleigh number $\Ra$.

\subsection{Fixed flux boundary conditions}
\label{ssec:neubnd}

In the opposite extreme, for Neumann BCs, we have $\dzT = \dztau = 1$,
and we bound $\dT$ from below using \eqref{eq:dTneupwlin}.  Since $b >
1$, in order for the lower bound $\BpwldT{\infty}(\delta,b) = 1 - b +
2 \delta \, b$ on $\dT$ to remain positive as $\R \to \infty$ and
hence $\delta \to 0$, we need $b - 1 = \bigO(\delta)$.  Thus following
\cite{OWWD02} we choose $b = 1 + c \, \delta$ and let $c$ take its
optimal value $\cinf = 1/2$, so that \eqref{eq:dTneupwlin} becomes
\begin{equation}
  \label{eq:Nudpwlneu}
  \Nu^{-1} = \dT \geq \BpwldT{\infty}(\delta,1+\cinf \delta) = 1 +
  (1+\delta/2) (2\delta - 1) = \frac{3}{2} \delta + 
    \delta^2 \sim \frac{3}{2} \delta .
\end{equation}
The condition on $\delta$ is as usual $\delta \leq \delta_c$, where
with $\dztau = 1$ and $b = 1+\delta/2$, the equation
\eqref{eq:deltaconstr} satisfied by $\delta_c$ takes the form
\begin{equation}
  \label{eq:deltacneu}
  \delta^4 = \frac{8}{\Reff} = 4 \frac{\delta}{1 + \delta/2}
  \R^{-1} \sim 4 \frac{\delta}{\R} 
\end{equation}
for large $\R$, for which $\delta \to 0$; and hence $\delta_c \sim
4^{1/3} \R^{-1/3}$.  Thus we have (using \eqref{eq:RRa})
\begin{align*}
  \Nu^{-1} = \dT & \geq \BpwldT{\infty}(\delta_c,1+\delta_c/2) \sim
  \frac{3}{2} \delta_c \sim \frac{3}{2^{1/3}} \R^{-1/3} , \\
  \Ra = \R \dT & \geq \R \, \BpwldT{\infty}(\delta_c,1+\delta_c/2)
  \sim \frac{3}{2^{1/3}} \R^{2/3} , 
\end{align*}
and so
\begin{equation}
  \label{eq:Nubndneu}
  \Nu \leq \tilde{\Nubnd}_{\text{pwl}} (\R) 
  \lesssim \frac{2^{1/3}}{3} \R^{1/3} \lesssim \sqrt{\frac{2}{27}} \Ra^{1/2} , 
\end{equation}
as in \cite{OWWD02}. 
Note the scaling $\Nu \leq C_1 \R^{1/3}$ in terms of the control
parameter $\R$, which translates to the usual scaling $\Nu \leq C_2
\Ra^{1/2}$.

\subsection{Mixed thermal boundary conditions with fixed Biot number}
\label{ssec:robbnd}

For general mixed (Robin) thermal BCs with fixed Biot number, we need
to estimate both $\dT$ and $\dzT$, using \eqref{eq:dzTrobpwlin} and
\eqref{eq:dTrobpwlin}, where $\dztau$ and $\dtau$ are given in terms
of $\Bi$ and $\delta$ by \eqref{eq:gammaBi}.  
The sufficient condition $\delta \leq \delta_c$ for $\tau_{\delta}$ to
be admissible, derived via the Cauchy-Schwarz estimates of
Section~\ref{ssec:CSestimates}, is that $\delta_c$ satisfies
\eqref{eq:deltaconstr}, which (substituting for $\dztau$ from
\eqref{eq:gammaBi}) here takes the form
\begin{equation}
  \label{eq:deltacrob}
  \dztau^2 \delta^4 = \frac{(1+2\Bi)^2}{4 (\delta+\Bi)^2} \delta^4 =
  \frac{8}{\Reff} = 8 \frac{b-1}{b} \R^{-1} .
\end{equation}
We shall see that in this general case with $0 < \Bi < \infty$,
depending on the relative sizes of $\delta$ and $\Bi$, the scaling of
the bounds behaves either as in the fixed temperature limit (for
$\delta \geq \Bi$) or the fixed flux limit (for $\delta \leq \Bi$);
but that for any $\Bi > 0$, the asymptotic scaling as $\R \to \infty$
is as for fixed flux boundary conditions:

\subsubsection{The fixed temperature problem $\Bi = 0$ as a singular limit:}
\label{sssec:singulardir}

Recall that for Dirichlet thermal boundary conditions $\Bi = 0$, we
have $\dT = \dtau = 1$, so that we obtain an upper bound on $\Nu$ for
\emph{any} $b > 0$ (there is no concern that the lower bound
$\BpwldT{0}$ on $\dT$ may become negative), and we can choose $b - 1 =
\bigO(1)$ for all $\delta$.  In this case $\Bi = 0$, though, $\dztau =
1/2\delta$ is not bounded above as $\R \to \infty$ ($\delta \to 0$),
and hence neither is $\dzT$; the growth in the (upper bound for) the
Nusselt number in the fixed temperature case with increasing control
parameter $\R = \Ra$ is due to that of the (non-dimensional) boundary
heat flux.

The situation is quite different for any nonzero Biot number $\Bi$:
since $0 < \delta \leq 1/2$, we have $0 \leq (1-2\delta)/2(\delta+\Bi)
= \dztau - 1 < 1/2\Bi$, so that now $\dztau$ \emph{is} bounded above
as $\delta \to 0$ for $\Bi > 0$.  On the other hand, $\dtau = 2\delta
\dztau$ is not bounded away from zero, so that since $b > 1$, to get a
positive value for the lower bound $\BpwldT{\Bi} = 1 - b + b\dtau$ for
sufficiently large $\R$ (small $\delta$), we need $b - 1 =
\bigO(\delta)$ for each fixed $\Bi > 0$.  Furthermore, we have $\dT
\to 0$ as $\R \to \infty$, so that (for sufficiently large $\R$) the
growth in the Nusselt number bound is due to the decrease in $\dT$,
the (non-dimensional) averaged temperature drop across the fluid,
rather than due to the growth in $\dzT$.  That is, for any $\Bi > 0$
the (asymptotic) behaviour and scaling is as in the fixed flux case;
the fixed temperature problem is a \emph{singular limit}.  (A similar
observation was made in the context of horizontal convection by
\cite{SKB04}.)

\subsubsection{Scaling regimes:}
\label{sssec:scregimes}

More precisely, the nature of the $\Nu$-$\R$ scaling depends on
whether $\delta \geq \Bi$ or $\delta \leq \Bi$, and hence on the value
of $\Bi$: 

For \emph{sufficiently large Biot number} (largely insulating
boundary) $\Bi \geq 1/2$, we always have $\delta \leq \Bi$.  Since for
such $\Bi$, $\dztau$ is approximately constant ($1 \leq \dztau < 1 +
1/2\Bi \leq 2$; compare $\dztau = 1$ for $\Bi = \infty$), we see from
\eqref{eq:deltacrob} that a sufficient admissibility condition for
$\taud$ is $\delta \leq \delta_c = \bigO(\Reff^{-1/4})$, as in the
fixed flux case.  We choose $b = 1 + c \delta \leq 3/2$ for some $c \leq 1$, so
$\dzT \leq 1 + b(\dztau - 1) \leq 5/2$ for all $\delta \leq 1/2$, and
there is no transition in scaling regimes; as in the fixed flux case,
for all sufficiently large $\Bi$ the growth in $\Nu$ is due to the
decrease in $\dT$.

For \emph{relatively small Biot number} (largely conducting boundary)
$\Bi < 1/2$, on the other hand, it is possible to have $\delta \geq
\Bi$ for low enough thermal driving, and thus distinct regimes
exist.  In particular, consider the case of small Biot number ($\Bi
\ll 1$, near the fixed temperature limit), where we can identify two
distinct scaling behaviours:
\begin{itemize}
\item ``Fixed temperature scaling'': For sufficiently small $\R$, we
  have $\delta \gg \Bi$, so that $\dztau \sim 1/2\delta$ and $\dtau
  \sim 1$,\footnote{Proceeding more carefully, for $\delta \geq \Bi$,
    we have $1/4\delta \leq \dztau = (1+2\Bi)/2(\delta + \Bi) \leq
    1/\delta$ and $1/2 \leq \dtau = \delta(1+2\Bi)/(\delta+\Bi) \leq
    1$} and the sufficiency condition \eqref{eq:deltacrob} is $\delta
  \leq \delta_c = \bigO(\Reff^{-1/2})$.  Since $\dtau$ is bounded
  below away from zero, so is the lower bound $\BpwldT{\Bi}(\delta,b)
  = 1 + b(\dtau - 1) \geq 1 - b/2$ on $\dT$ for any fixed $b < 2$.
  Thus we may obtain a bound on $\Nu$ in this regime by choosing any
  $b \in (1,2)$, and by comparison with the fixed temperature problem,
  it is sufficient to choose $b - 1 = \bigO(1)$, in which case we have
  $\delta_c = \bigO([(b-1)/b\R]^{1/2}) = \bigO(\R^{-1/2})$.  While
  $\dT = \bigO(1)$ (so that $\Ra = \bigO(\R)$), we have that
  $\dzT \leq \BpwldzT{\Bi}(\delta,b) = 1 + b(\dztau - 1) =
  \bigO(b/\delta)$ grows as $\delta^{-1}$.  Thus clearly when $\Bi \ll
  1$, for sufficiently small but increasing $\R$, the scaling
  properties are as in the fixed temperature case, and the growth in
  $\Nu = \dzT/\dT$ is driven by that of $\dzT$.

  As the driving $\R$ increases, $\delta$ decreases, and eventually
  becomes less than the Biot number $\Bi$; based on the fixed
  temperature scaling $\delta = \bigO(\R^{-1/2}) = \bigO(\Ra^{-1/2})$
  the transition at $\delta = \Bi$ occurs when $\Bi = \bigO(\Ra^{-1/2})$,
  or $\Ra = \bigO(\Bi^{-2})$.

\item ``Fixed flux scaling'': Once the ``boundary layer thickness''
  $\delta$ has decreased below $\Bi > 0$ for increasing $\R$, we enter
  another regime (which does not exist in the fixed temperature case
  $\Bi = 0$), in which for fixed $\Bi$ the growth in $\dztau$
  saturates, while $\dtau = \bigO(\delta)$ decreases.  Asymptotically
  for $\delta \ll \Bi$, we have $\dztau \sim (1+2\Bi)/2\Bi =
  \dztau_{\text{max}}(\Bi)$, while $\dtau \sim \delta(1+2\Bi)/\Bi$,
  and for each fixed $\Bi > 0$ the behaviour is now as if we had
  Neumann thermal BCs.\footnote{More precisely, for $\delta \leq \Bi$,
    we have $\dztau_{\text{max}}/2 = (1+2\Bi)/4\Bi \leq \dztau <
    (1+2\Bi)/2\Bi = \dztau_{\text{max}}$, and $\delta
    \dztau_{\text{max}} = \delta(1+2\Bi)/2\Bi \leq \dtau <
    \delta(1+2\Bi)/\Bi = 2\delta \dztau_{\text{max}}$.}

  More generally, for $0 < \Bi \leq 1/2$ and decreasing $\delta \leq
  \Bi$, we have $\dztau = \bigO(\Bi^{-1})$ and $\dtau =
  \bigO(\delta/\Bi)$.  In order for the lower bound $\BpwldT{\Bi} = 1
  - b + b\dtau$ on $\dT$ to remain positive as $\delta \to 0$, we must
  choose $b = 1 + \bigO(\delta/\Bi)$, so that $\dzT \leq \BpwldzT{\Bi}
  = 1 + b(\dztau - 1) = \bigO(\Bi^{-1})$ saturates, while $\dT \geq
  \bigO(\delta/\Bi)$; hence the growth in $\Nu$ is now due to the
  decay in $\dT$, as in the fixed flux case.  In this regime the
  scaling behaviours are $\Ra \geq \bigO(\delta \R/\Bi)$, $\Reff =
  \bigO(\Bi \R/\delta)$ and $\delta = \bigO(\dztau^{-1/2}
  \Reff^{-1/4}) = \bigO(\Bi^{1/3} \R^{-1/3}) = \bigO(\Ra^{-1/2})$;
  more precise asymptotic statements are given below, with weaker, but
  rigorous results in Appendix~\ref{app:rigbounds}.

\end{itemize}

\subsubsection{Asymptotic scaling of bounds for $0 < \Bi < \infty$:}
\label{sssec:biotbnds}

Having outlined the behaviour in the different regimes, we
here derive the scaling of the bound on the Nusselt number in the limit
of large driving, $\R \to \infty$, so that $\delta \ll 1$ and $\delta
\ll \Bi$, deferring a more detailed discussion of scaling behaviour in
the different regimes using this Cauchy-Schwarz analysis, and a
comparison with numerical solutions for piecewise linear backgrounds,
to \cite{WiGa08prep}.

In the light of the above discussion, for $\delta \ll \Bi$ we must
choose $b = 1 + c\, \delta$, where the optimal value of $c$ turns out to
be
\begin{equation}
  \label{eq:cBi}
  \cBi = \frac{1+2\Bi}{4\Bi}  .
\end{equation}
Using this optimal choice of $b$, the lower bound
\eqref{eq:dTrobpwlin} on $\dT$ becomes
\begin{align}
  \dT \geq \BpwldT{\Bi} (\delta,1+\cBi \delta) & = - \cBi \delta +
  (1+\cBi \delta) \frac{\delta (1+2\Bi)}{\delta + \Bi} \nonumber \\
  \label{eq:dTdpwlrob}
  & = \frac{\delta (1+2\Bi)}{\delta + \Bi} \frac{3 + 2\delta}{4} 
  \sim \frac{3}{4} \frac{\delta (1+2\Bi)}{\Bi} ,
\end{align}
while similarly, the upper bound \eqref{eq:dzTrobpwlin} is
\begin{align}
  \dzT \leq \BpwldzT{\Bi} (\delta, 1+\cBi \delta) & = - \cBi \delta +
  (1 + \cBi \delta) \frac{1+2\Bi}{2(\delta + \Bi)} \nonumber \\
  \label{eq:dzTdpwlrob}
  & = \frac{1+2\Bi}{2(\delta + \Bi)} \left[ 1 + \frac{\delta}{4\Bi} (
    1 - 2\delta) \right] 
  \sim \frac{1+2\Bi}{2\Bi} ,
\end{align}
so that an upper bound on the Nusselt number for admissible $\delta
\ll \Bi$ is
\begin{equation}
  \label{eq:Nudpwlrob}
  \Nu = \frac{\dzT}{\dT} \leq \Nupwl{\Bi} (\delta, 1 + \cBi \delta) =
  \frac{1}{2\delta} \frac{4 + \delta(1-2\delta)/\Bi}{3 + 2\delta}
  \sim \frac{2}{3\delta} ; 
\end{equation}
compare \eqref{eq:Nudpwldir} and \eqref{eq:Nudpwlneu}. 

Observe that the width $\delta_{\text{BL}}$ of the thermal boundary
layer is often related to the Nusselt number via $\delta_{\text{BL}} =
(2\Nu)^{-1}$ (\cite{NiSr06}); our high-$\R$ result for the piecewise
linear background, $\delta \sim (3\Nu/2)^{-1}$ for $\Bi > 0$ (or
$\delta \sim (4\Nu/3)^{-1}$ for $\Bi = 0$), may be interpreted as a
systematic statement of such a boundary layer model.

Returning to the computation of asymptotic bounds, we note from
\eqref{eq:cBi} that for $\Bi \geq 1/2$, $\cBi = 1/2 + 1/4\Bi \leq 1$,
while for $\Bi \leq 1/2$, $\cBi \delta = (1+2\Bi) \delta/4\Bi \leq
\delta/2\Bi$, so that whenever $\delta \ll \min(\Bi,1)$ we have $\cBi
\delta \ll 1$; consequently $b = 1 + \cBi \delta \sim 1$ and $\Reff =
b \, \R /(b-1) \sim \R/\cBi \delta$.  In this case the condition
\eqref{eq:deltacrob} is thus
\begin{equation}
  \label{eq:deltacrob2}
  \delta^4 = 32 \frac{(\delta+\Bi)^2}{(1+2\Bi)^2} \Reff^{-1} 
  \sim 32 \frac{\Bi^2}{(1+2\Bi)^2} \frac{1+2\Bi}{4\Bi} \delta
  \R^{-1} = 8 \frac{\Bi}{1+2\Bi} \delta \R^{-1} ,
\end{equation}
or $\delta_c \sim 2 \Bi^{1/3} (1+2\Bi)^{-1/3} \R^{-1/3}$.
Substituting into the above bounds, we have
\begin{align}
  \label{eq:NuRbndrob}
  \Nu \leq \Nupwl{\Bi} (\delta_c, 1 + \cBi \delta_c) & \sim
  \frac{2}{3\, \delta_c} \sim \frac{1}{3} \left( \frac{1+2\Bi}{\Bi}
  \right)^{1/3} \R^{1/3} , \\
  \label{eq:RaRbndrob}
  \Ra = \R \dT \geq \R \BpwldT{\Bi} (\delta_c, 1 + \cBi \delta_c) & \sim
  \frac{3}{4} \frac{1+2\Bi}{\Bi} \, \delta_c \, \R \sim \frac{3}{2} \left(
    \frac{1+2\Bi}{\Bi} \right)^{2/3} \R^{2/3} ,
\end{align}
so that we obtain a bound on the asymptotic scaling as $\R \to \infty$
of the Nusselt number with the Rayleigh number whenever $\Bi > 0$:
\begin{equation}
  \label{eq:Nubndrob}
  \Nu \lesssim \frac{1}{3} \left( \frac{1+2\Bi}{\Bi} \right)^{1/3}
  \sqrt{\frac{2}{3}} \left( \frac{\Bi}{1+2\Bi} \right)^{1/3} \Ra^{1/2} =
  \sqrt{\frac{2}{27}} \Ra^{1/2} ,
\end{equation}
\emph{independent of the Biot number}.  Observe in particular, by
comparison with \eqref{eq:Nubndneu}, that the prefactor $\sqrt{2/27}$
is the same as for the fixed flux problem.


\section{Conclusions}
\label{sec:concl}

In formulating the energy identities and bounding problem for the
Rayleigh-B\'enard model with finite Prandtl number and general
thermal BCs at the upper and lower boundaries of the fluid, we have
demonstrated that the fixed temperature and fixed flux extremes may
indeed be treated as special cases of a more general model, within
which one can rigorously prove energy boundedness and bounds on
convective heat transport, and obtain asymptotic scaling results; we
expect that such an approach may be applicable to other related
convection problems.

While the scaling of these analytical bounds on the $\Nu$--$\Ra$
relationship remains well above that observed experimentally or in
direct numerical simulations, some of the qualitative conclusions may
be instructive.  Of particular interest is that---at least for the
piecewise linear backgrounds $\taud(z)$ treated here---while for each
fixed $\R$ the bounds depend smoothly on $\Bi$ for $0 \leq \Bi \leq
\infty$, the asymptotic $\R \to \infty$ scaling of the bound for any
nonzero Biot number is as for the $\Bi = \infty$ fixed flux problem.
That is, the limits $\Bi \to 0$ and $\R \to \infty$ do not commute:
fixed temperature conditions in fact form a singular limit.

Furthermore, the bounding calculation indicates the existence of two
distinct scaling behaviours for sufficiently small nonzero $\Bi$; it
would be of interest to observe these in fixed Biot number direct
numerical simulations: For small Rayleigh number $\Ra$, there is a
``fixed temperature scaling regime'' in which the usual assumption of
Dirichlet thermal BCs is approximately valid, as the growth in the
convective heat transport measured by $\Nu$ is largely due to the
increase in the averaged boundary heat flux $\dzT$.  As the control
parameter $\R$, and hence $\Ra$, increases, a transition occurs when
the ``boundary layer width'' $\delta$ becomes comparable to $\Bi$ (in
our calculations this occurs for $\Ra = \bigO(\Bi^{-2})$), beyond
which a ``fixed flux scaling regime'' is entered, in which further
increases in $\Nu$ are driven by decreases in the averaged temperature
drop $\dT$.  This observation provides mathematical support for the
heuristic argument that when the Nusselt number is sufficiently high,
the boundaries act effectively as insulators.

In the three-dimensional simulations of \cite{VeSr08} in cylindrical
geometry with perfectly insulating sidewalls and a perfectly
conducting upper boundary, replacing the lower fixed temperature BCs
with fixed flux conditions was observed to have little effect on the
heat transport for a given Rayleigh number for sufficiently small $\Ra
\lesssim 10^9$, and to decrease the transport for $\Ra > 10^9$.  In
constrast, the two-dimensional, horizontally periodic computations of
\cite{JoDo08} showed essentially identical heat transport for fixed
temperature and fixed flux BCs at both upper and lower plates for
$10^7 \lesssim \Ra \lesssim 10^{10}$.  In this context we observe that
the prefactor in our asymptotic analytical bound $\Nu \leq C\,
\Ra^{1/2}$ increases from $C_0 = 3\sqrt{6}/32 \approx 0.230$ to
$C_{\Bi} = C_{\infty} = \sqrt{2/27} \approx 0.272$ for $\Bi > 0$; that
is, within the framework of our upper bounding calculations with
piecewise linear background it appears that the estimates on the heat
transport \emph{increase} when the boundaries are not perfectly
conducting.  It remains to determine whether this increase is an
artifact of the choice of background $\tau(z)$ or of the background
flow bounding approach in general.  A further consideration is how the
presence of finite width conducting plates (see Part 2 of this work)
modifies conclusions obtained with the fixed Biot number
simplification.


\vspace{2ex}

\subsubsection{Acknowledgments}
\label{sec:ack}

I would like to thank Charlie Doering, Jian Gao, Jesse Otero and
Jean-Luc Thiffeault for useful discussions concerning this work.  This
research was partially supported by grants from the Natural Sciences
and Engineering Research Council of Canada (NSERC).


\appendix

\section{Boundedness of $\|T\|^2$ and $\|\bu\|^2$}
\label{app:bounded}

For completeness of the rigorous argument, we show the uniform
boundedness of the temperature and velocity fields, which we may state
as a theorem:
\begin{thm}
\label{thm:unifbnd}
The velocity field $\bu$ and temperature field $T$ satisfying
\eqref{eq:Bous}--\eqref{eq:heatfluid} for finite Prandtl number, $0 <
\Prandtl < \infty$, and for no-slip velocity boundary conditions and
thermal boundary conditions of general Biot number $0 \leq \Bi \leq
\infty$, are uniformly bounded in $\Ltwo$.
\end{thm}
\begin{remark}
Such boundedness has already been shown for Rayleigh-B\'enard
convection with fixed temperature BCs by \cite{Kers01}, following the
underlying approach introduced by \cite{DoCo92} (based on an idea of
\cite{Hopf41}) in the context of shear flow.  However, in both of
these cases the Dirichlet boundary conditions allow ready control of
the indefinite term in real space; since for general thermal BCs we
are not assured control of $\theta$ at the fluid boundaries, in our
proof instead we use incompressibility and Fourier space estimates
based on those of \cite{OWWD02}.
\end{remark}

\begin{proof}
  We begin the demonstration of Theorem~\ref{thm:unifbnd} by reviewing
  the basic problem formulation and identities: With thermal boundary
  conditions imposed at the upper and lower limits of the fluid, we
  consider horizontally periodic temperature and velocity fields
  $T(\bx,t)$ and $\bu(\bx,t)$ satisfying
  \eqref{eq:Bous}--\eqref{eq:heatfluid}, where $\bu$ satisfies
  incompressibility and no-slip BCs.  Choosing a background
  temperature field $\tau(z)$ which satisfies the given thermal
  boundary conditions, we define $\bv(\bx,t)$ and $\theta(\bx,t)$ via
  the decomposition \eqref{eq:decomp}, $\bu(\bx,t) = \bv(\bx,t)$,
  $T(\bx,t) = \tau(z) + \theta(\bx,t)$, and thus obtain the evolution
  equations \eqref{eq:vL2}--\eqref{eq:thL2} for their $\Ltwo$ norms:
\begin{align}
  \label{eq:AvL2}
  \frac{1}{2 \Prandtl \R} \frac{d}{dt} \Ltf{\bv}^2 & = - \frac{1}{\R}
  \Ltf{\del \bv}^2 + \intf w \theta , \\
  \label{eq:AthL2}
  \frac{1}{2} \frac{d}{dt} \Ltf{\theta}^2  & = - \Ltf{\del \theta}^2 +
  A \zlim{\have{\theta \theta_z}} - \intf \theta_z
  \tau' + A \zlim{\have{\theta} \tau'} - \intf w \theta \tau' .
\end{align}
We form the linear combination \eqref{eq:AvL2} + $\muwt
\cdot$\eqref{eq:AthL2}, where the weight $\muwt$ will be chosen later:
\begin{align}
  \frac{1}{2} \frac{d}{dt} \left[ \Ltf{\theta}^2 + \frac{\muwt}{\Prandtl
      \R} \Ltf{\bv}^2 \right] & = - \frac{\muwt}{\R} \Ltf{\del \bv}^2 +
  \muwt \intf w \theta - \intf \tau' w \theta -
  \Ltf{\del \theta}^2 \nonumber \\
  \label{eq:AL2evol}
    & \qquad \ - \intf \tau' \theta_z + A \zlim{\have{\theta \theta_z}} + A
    \zlim{\have{\theta} \tau'} .
\end{align}

We choose, as before, a piecewise linear $\tau(z)$, defined as in
\eqref{eq:taudelta} and figure~\ref{fig:pwlin_biot} for $0 < \delta \leq
1/2$:
\begin{equation}
  \label{eq:Ataudelta}
  \tau(z) = \left \{
    \begin{array}{ll}
      \avetau - \dztau (z - \delta) , & 0 \leq z \leq \delta , \\
      \avetau , & \delta < z < 1 - \delta , \\
      \avetau - \dztau(z - 1 + \delta) , & 1-\delta \leq z \leq 1 ; \\
    \end{array}
  \right. 
\end{equation}
As in Section~\ref{ssec:pwlinear}, we find using the thermal BCs that
\eqref{eq:gammaBi} $\dztau = -\tau'(0) = -\tau'(1) =
(1+2\Bi)/2(\delta+\Bi)$, while for $\Bi < \infty$ we also have
$\avetau = 1/2$; for now we defer the (at present arbitrary) choice of
$\avetau$ in the fixed flux case $\Bi = \infty$.

\subsection{Estimates independent of thermal BCs:} 
\label{sapp:estindep}

The estimates on the indefinite quadratic terms are performed in
Fourier space using the definition \eqref{eq:wthk}, while other terms
are readily controlled in real space; thus we split the dissipative
terms as
\begin{equation}
  \label{eq:Agradth}
  \Ltf{\del \theta}^2 = \left( \half  + \half\right) \intf |\del
  \theta|^2 = \half \Ltf{\del \theta}^2 + \half A \sum_{\bk}
  \int_0^1 \left( k^2 |\thk|^2 + |\D \thk|^2 \right) \, dz
\end{equation}
and, using incompressibility, 
\begin{equation}
  \label{eq:Agradv}
  \Ltf{\del \bv}^2 
  \geq \half \Ltf{\del \bv}^2 + \half A \sum_{\bk} \int_0^1 
  \left( k^2 |\whk|^2 + 2 |\D \whk|^2 + \frac{1}{k^2} |\D^2 \whk|^2
  \right) \, dz ,
\end{equation}
with equality for two-dimensional flows.

We bound $\intf w \theta$ using estimates of the form
\eqref{iq:whkthkz2}, which imply that for $\bk \not= \bO$, any $0 <
\delta \leq 1/2$ and any $p, q > 0$,
\begin{equation}
  \label{iq:Aiwhkthkd}
  \begin{split}
    \left( \int_0^{\delta} |\whk \conj{\thk}| \, dz +
      \int_{1-\delta}^1 |\whk \conj{\thk}| \, dz \right) & \\
    \leq \frac{\delta^2}{4 \sqrt{2}} & \left[
      p \| \D \whk \|^2 + \frac{1}{p} \| \D
    \thk \|^2 + \frac{q}{k^2} \| \D^2 \whk \|^2 +
    \frac{k^2}{q} \| \thk \|^2 \right] .
  \end{split}
\end{equation}
In particular, using \eqref{iq:Aiwhkthkd} with $\delta = 1/2$,
recalling $\wh_{\bO} = 0$,
\begin{align}
  \frac{1}{A} \intf w \theta & = \sum_{\bk} \int_0^1 \Real{\whk
    \conj{\thk}} \, dz = \frac{1}{2} \sum_{\bk} \int_0^1 ( \whk \conj{\thk} +
  \conj{\whk} \thk )  \, dz \nonumber \\
  \label{iq:Aiwth}
  & \leq \frac{1}{16 \sqrt{2}} \sum_{\bk} \left[
      a_1 \| \D \whk \|^2 + \frac{1}{a_1} \| \D
    \thk \|^2 + \frac{a_2}{k^2} \| \D^2 \whk \|^2 +
    \frac{k^2}{a_2} \| \thk \|^2 \right] .
\end{align}

Similarly, using the definition \eqref{eq:Ataudelta} of $\tau(z)$, as in
\eqref{iq:indef} we have
\begin{align}
  \frac{1}{A} \left| \intf \tau' w \theta \right| & \leq \frac{1}{2}
  \sum_{\bk} \left|  \int_0^1 \tau' ( \whk \conj{\thk} + \conj{\whk} \thk )
    \, dz \right| \nonumber \\ 
  \label{iq:Aiwthtau}
  & \leq \frac{\dztau \delta^2}{4 \sqrt{2}} \sum_{\bk} \left[
    a_3 \| \D \whk \|^2 + \frac{1}{a_3} \| \D
    \thk \|^2 + \frac{a_4}{k^2} \| \D^2 \whk \|^2 +
    \frac{k^2}{a_4} \| \thk \|^2 \right] .
\end{align}

The other two terms in \eqref{eq:AL2evol} containing $\tau'$ may be
estimated directly, since for any $a_5 > 0$, we have (also using
$\|\theta_z\| \leq \|\del \theta\|$)
\begin{align}
  - \intf \tau' \theta_z + A \zlim{\tau' \have{\theta}} 
  & = \dztau A \left( \int_0^{\delta} \have{\theta}_z \, dz +
    \int_{1-\delta}^1 \have{\theta}_z \, dz \right) - \dztau A
  \int_0^1 \have{\theta}_z \, dz \nonumber \\
  & = -\dztau A \int_{\delta}^{1-\delta} \have{\theta}_z \, dz =
  -\dztau \int_{\delta}^{1-\delta} \iint_A \theta_z \, dx \, dy \, dz
  \nonumber \\
  & \leq \dztau A^{1/2} \sqrt{1-2\delta} \left(
    \int_{\delta}^{1-\delta} \iint_A \theta_z^2 \, dx   \, dy \, dz
  \right)^{1/2} \nonumber \\
  \label{iq:Ataupthz}
  & \leq \dztau A^{1/2} \| \theta_z \| 
  \leq \half \left( a_5 \dztau^2 A + \frac{1}{a_5} \| \del \theta
    \|^2 \right) . 
\end{align}

Substituting \eqref{eq:Agradth}--\eqref{iq:Ataupthz} into
\eqref{eq:AL2evol}, and collecting like terms, we thus obtain
\begin{align}
  \frac{1}{2} \frac{d}{dt} \left[ \Ltf{\theta}^2 + \frac{\muwt}{\Prandtl
      \R} \Ltf{\bv}^2 \right] 
  & \leq \half a_5 \dztau^2 A - \frac{\muwt}{2\R} \Ltf{\del
    \bv}^2 - \half \left( 1 - \frac{1}{a_5} \right) \Ltf{\del
    \theta}^2 + A \zlim{\have{\theta \theta_z}} \nonumber \\
  & \quad - A \sum_{\bk} \frac{\muwt}{2\R} k^2 \|\whk\|^2
  - A \sum_{\bk} \left( \frac{\muwt}{\R} - \frac{a_1
      \muwt}{16 \sqrt{2}} - \frac{a_3 \dztau \delta^2}{4\sqrt{2}}
  \right) \| \D \whk \|^2 \nonumber \\
  & \quad - A \sum_{\bk} \left( \frac{\muwt}{2\R} - \frac{a_2
      \muwt}{16 \sqrt{2}} - \frac{a_4 \dztau \delta^2}{4\sqrt{2}}
  \right) \frac{1}{k^2} \| \D^2 \whk \|^2 \nonumber \\
  & \quad - A \sum_{\bk} \left( \half - \frac{\muwt}{16 \sqrt{2} \, a_2} -
    \frac{\dztau \delta^2}{4 \sqrt{2} \, a_4} \right) k^2 \| \thk \|^2
  \nonumber \\
  \label{iq:AL2full}
  & \quad - A \sum_{\bk} \left( \half - \frac{\muwt}{16 \sqrt{2} \, a_1} -
    \frac{\dztau \delta^2}{4 \sqrt{2} \, a_3} \right) \| \D \thk \|^2 .
\end{align}
At this point we are free to choose the constants $a_1$--$a_5$, as
well as $\muwt$ and $\delta$.  It is convenient to begin by selecting
$a_1$--$a_4$ so that $\muwt/16\sqrt{2}\, a_{1,2} = \dztau
\delta^2/4\sqrt{2}\, a_{3,4} = 1/4$, and then, after substitution, to
choose $\muwt$ and $\delta$ to satisfy $a_2\muwt/16\sqrt{2} = a_4 \dztau
\delta^2/4\sqrt{2} = \muwt/4\R$.  This gives
\begin{equation}
  \label{eq:a1tomu}
  a_1 = a_2 = \frac{\muwt}{4\sqrt{2}}, \qquad a_3 = a_4 = \frac{\dztau
  \delta^2}{\sqrt{2}}, \qquad \muwt = \frac{32}{\R}, 
\end{equation}
and $\dztau^2 \delta^4/8 = \muwt/4\R = 8/\R^2$, so that $\dztau \delta^2
= 8/\R$.  Choosing, furthermore, $a_5 = 2$, and substituting,
\eqref{iq:AL2full} becomes
\begin{align}
  \frac{1}{2} \frac{d}{dt} \left[ \Ltf{\theta}^2 + \frac{32}{\Prandtl
      \R^2} \Ltf{\bv}^2 \right] 
  & \leq \dztau^2 A - \frac{16}{\R^2} \Ltf{\del
    \bv}^2 - \frac{1}{4} \Ltf{\del
    \theta}^2 + A \zlim{\have{\theta \theta_z}} \nonumber \\
  & \quad - A \sum_{\bk} \frac{16}{\R^2} \left( k^2 \|\whk\|^2
    + \| \D \whk \|^2 \right) \nonumber \\
  \label{iq:AL2full2}
  & \leq \dztau^2 A - \frac{16}{\R^2} \Ltf{\del
    \bv}^2 - \frac{1}{4} \Ltf{\del
    \theta}^2 + A \zlim{\have{\theta \theta_z}} .
\end{align}

\subsection{Poincar\'e and related inequalities}
\label{sapp:poincare}

It remains to establish Poincar\'e-like inequalities controlling
$\Ltf{\bv}^2$ and $\Ltf{\theta}^2$ by $\Ltf{\del \bv}^2$ and
$\Ltf{\del \theta}^2$, respectively.  In the following we consider
only functions periodic in the horizontal directions, on $(x,y) \in
[0,L_x]\times[0,L_y]$, so that further discussion of ``boundary
conditions'' refers to the vertical boundaries at $z = 0$ and $z = 1$.

For nonzero functions $\psi$ which vanish at the vertical boundaries,
we have the Poincar\'e inequality $\| \del \psi \|^2/\|\psi\|^2 \geq
\lameig_D$, where $\lameig_D = \pi^2$ is the lowest 
eigenvalue\footnote{The use of $\lameig$ (or $\lameig_{D,N,R}$) to
  represent eigenvalues only in this Appendix~\ref{app:bounded}
  should not be confused with the conductivity $\lambda$ elsewhere.}
of $-\lap$ on this domain $[0,L_x]\times[0,L_y]\times[0,1]$ with
homogeneous Dirichlet BCs at $z = 0,1$.  Applying this inequality to
the three components of the velocity field with no-slip BCs, we have
\begin{equation}
  \label{iq:APoinv}
  - \Ltf{\del \bv}^2 \leq - \lameig_D \Ltf{\bv}^2 = - \pi^2
  \Ltf{\bv}^2 .
\end{equation}
In order to establish an analogous inequality for $-\Ltf{\del
  \theta}^2$, and to find $\dztau$ and $\delta$ from $\dztau \delta^2
= 8/\R$, we need to consider the different thermal boundary conditions
separately:

\subsubsection{Inequalities for fixed temperature BCs:}
\label{ssapp:poindir}

For functions $\theta$ satisfying Dirichlet BCs $\theta|_{z=0,1} = 0$
we have, as discussed above,
\begin{equation}
  \label{iq:APointhdir}
  - \Ltf{\del \theta}^2 \leq - \lameig_D \Ltf{\theta}^2
\end{equation}
for $\lameig_D = \pi^2$.  Since in this case the term $A
\zlim{\have{\theta} \tau'}$ in \eqref{eq:AL2evol} vanishes, we can
improve upon \eqref{iq:Ataupthz} to give, for $0 < \delta \leq 1/2$, 
\begin{align}
  - \intf \tau' \theta_z & = \dztau \left( \int_0^{\delta} \iint_A
    \theta_z \, dx \, dy \, dz + \int_{1-\delta}^1  \iint_A
    \theta_z \, dx \, dy \, dz \right) \nonumber \\
  & \leq \dztau A^{1/2} \delta^{1/2} \left[ \left( \int_0^{\delta}
      \iint_A \theta_z^2 \, dx \, dy \, dz \right)^{1/2} + \left(
      \int_{1-\delta}^1 \iint_A \theta_z^2 \, dx \, dy \, dz
    \right)^{1/2} \right] \nonumber \\
  \label{iq:Ataupthzdir}
  & \leq \half \left( 2 a_5 \dztau^2 \delta A + \frac{1}{a_5}
    \Ltf{\del \theta}^2 \right) .
\end{align}
Using \eqref{iq:Ataupthzdir} instead of \eqref{iq:Ataupthz}, we obtain
an equation similar to \eqref{iq:AL2full} in which the first term on
the right-hand side is $a_5 \dztau^2 \delta A$, and choose $a_1$--$a_5$,
$\muwt$ and $\delta$ as before.  In the fixed temperature case we have
$\dztau = 1/2\delta$, so that the condition on $\delta$ becomes
$\dztau \delta^2 = \delta/2 = 8/\R$, or $\delta = 16/\R$, $\dztau =
\R/32$.  Hence the equivalent of \eqref{iq:AL2full2} becomes
\begin{align}
  \frac{1}{2} \frac{d}{dt} \left[ \Ltf{\theta}^2 + \frac{32}{\Prandtl
      \R^2} \Ltf{\bv}^2 \right] 
  & \leq 2 \dztau^2 \delta A - \frac{16}{\R^2} \Ltf{\del
    \bv}^2 - \frac{1}{4} \Ltf{\del
    \theta}^2 + A \zlim{\have{\theta \theta_z}} \nonumber \\
  & = \frac{\R}{32} A - \frac{1}{4} \Ltf{\del \theta}^2 -
  \frac{16}{\R^2} \Ltf{\del \bv}^2  \nonumber \\
  & \leq \frac{\R}{32} A - \frac{\lameig_D}{4} \left[ \Ltf{\theta}^2 +
    \frac{64}{\R^2} \Ltf{\bv}^2 \right] \nonumber \\
  \label{iq:AL2full4dir}
  & \leq \frac{\R}{32} A - \frac{\lameig_D}{4} \min(2\Prandtl,1)
  \left[ \Ltf{\theta}^2 + \frac{32}{\Prandtl \R^2} \Ltf{\bv}^2 \right] ,
\end{align}
where we used $\zlim{\have{\theta \theta_z}} = 0$ for Dirichlet
thermal BCs, and the Poincar\'e inequalities
\eqref{iq:APoinv}--\eqref{iq:APointhdir}.

\subsubsection{Inequalities for fixed flux BCs:}
\label{ssapp:poinneu}

When $\theta$ satisfies Neumann BCs $\theta_z|_{z=0,1} = 0$, we again
have $\zlim{\have{\theta \theta_z}} = 0$.  In addition, $\dztau = 1$,
and $\delta$ is found from $\dztau \delta^2 = \delta^2 = 8/\R$, so
$\delta = \sqrt{8/\R}$; hence \eqref{iq:AL2full2} becomes 
\begin{equation}
  \label{iq:AL2full3neu}
  \frac{1}{2} \frac{d}{dt} \left[ \Ltf{\theta}^2 + \frac{32}{\Prandtl
      \R^2} \Ltf{\bv}^2 \right] 
  \leq A - \frac{16}{\R^2} \Ltf{\del \bv}^2 - \frac{1}{4}
  \Ltf{\del \theta}^2  .
\end{equation}
In this case, while the Poincar\'e inequality \eqref{iq:APoinv} holds
for the velocity field as before, the temperature field requires a bit
more care, since under Neumann thermal BCs the equations are
invariant under $\theta \mapsto \theta + $constant, and we have no
immediate control of $\Ltf{\theta}^2$ by $\Ltf{\del \theta}^2$.
However, in general for nonzero functions $\psi$ with \emph{mean
  zero}, we have $\|\del \psi\|^2/\|\psi\|^2 \geq \lameig_N$, where
$\lameig_N = \pi^2$ is the lowest \emph{nonzero} eigenvalue of $-\lap$
on this domain with homogeneous Neumann BCs at $z = 0, 1$; and we can
satisfy the additional condition on the mean (in the fixed flux case
only) by exploiting the remaining freedom in the definition
\eqref{eq:Ataudelta} of the background $\tau(z)$:

Since for fixed flux thermal BCs the flux out at the top of the fluid
exactly balances the flux in at the bottom, the total heat content is
preserved over time; more precisely, letting $\aveT = A^{-1} \intf T$
be the mean temperature over the fluid, in this case \eqref{eq:tfave}
becomes $d \aveT/dt = 0$.  Thus we may choose the
(previously arbitrary) average $\avetau$ of $\tau(z)$ to be the
(constant) average temperature, $A^{-1} \intf \tau = \avetau = \aveT$,
thereby completing the definition \eqref{eq:Ataudelta} for $\Bi =
\infty$.  Hence by construction the perturbation $\theta$ has mean zero,
$\intf \theta = 0$, so that we have the inequality
\begin{equation}
  \label{iq:APointhneu}
  - \Ltf{\del \theta}^2 \leq - \lameig_N \Ltf{\theta}^2 .
\end{equation}

Substituting \eqref{iq:APoinv} and \eqref{iq:APointhneu} into
\eqref{iq:AL2full3neu} now gives 
\begin{align}
  \frac{1}{2} \frac{d}{dt} \left[ \Ltf{\theta}^2 + \frac{32}{\Prandtl
      \R^2} \Ltf{\bv}^2 \right] 
  & \leq A  - \frac{1}{4} \left[ \lameig_N \Ltf{\theta}^2 + \frac{64}{\R^2}
  \lameig_D \Ltf{\bv}^2 \right] \nonumber \\
  \label{iq:AL2full4neu}
  & \leq A  - \frac{\lameig_N}{4} \min \left( 2\Prandtl
    \frac{\lameig_D}{\lameig_N}, 1 \right) \left[ \Ltf{\theta}^2 +
    \frac{32}{\Prandtl \R^2} \Ltf{\bv}^2 \right] .
\end{align}

\subsubsection{Inequalities for mixed thermal BCs:}
\label{ssapp:poinrob}

For thermal BCs with general Biot number $0 < \Bi < \infty$, we do not
have such simple expressions for $\delta$ and $\dztau$, but as before
we define $\delta = \delta(\R,\Bi)$ by $\dztau(\delta) \delta^2 = (1 +
2\Bi) \delta^2/2(\delta + \Bi) = 8/\R$.  Weakening the inequality
\eqref{iq:AL2full2} by using an upper bound for $\dztau$ (valid for
$\Bi > 0$) uniform in $\delta > 0$, we have in this case 
\begin{equation}
  \label{iq:AL2full3rob}
  \frac{1}{2} \frac{d}{dt} \left[ \Ltf{\theta}^2 + \frac{32}{\Prandtl
      \R^2} \Ltf{\bv}^2 \right] 
  \leq \frac{(1+2\Bi)^2}{4\Bi^2} A - \frac{16}{\R^2} \Ltf{\del
    \bv}^2 - \frac{1}{4} \Ltf{\del \theta}^2  + A \zlim{\have{\theta
      \theta_z}} .
\end{equation}
Observe that for $\Bi \in (0,\infty)$, the boundary term in
\eqref{iq:AL2full3rob} does not vanish in general; in lieu of a
Poincar\'e inequality, we exploit this fact in the following to
control $\| \theta\|^2$:

For $0 < \Bi < \infty$, it is straightforward to verify that for
nonzero functions $\psi$ on a general domain $\Omega$, the stationary
values of the ratio
\begin{equation*}
  \lameig = \frac{\intO |\del \psi|^2 + \intdO \Bi^{-1}
    \psi^2}{\intO \psi^2} 
\end{equation*}
are the eigenvalues of $-\lap$ on $\Omega$ with homogeneous Robin BCs
on $\partial \Omega$, and are attained at the corresponding
eigenfunctions; that is, at the solutions of
\begin{equation*}
  - \lap \psi = \lameig \psi \ \ \text{in}\ \ \Omega, \qquad \unitn
  \cdot \del \psi + \Bi^{-1} \psi = 0 \ \ \text{on}\ \ \partial
  \Omega ,
\end{equation*}
so that the lowest eigenvalue $\lameig_{R,\Omega} > 0$ minimizes the
above ratio.
Specializing to our fluid domain $[0,L_x]\times[0,L_y]\times[0,1]$ and
evaluating the boundary term, we
have similarly that for nonzero horizontally periodic functions
$\psi$, we have 
\begin{equation*}
  \frac{\intf |\del \psi |^2 + \Bi^{-1} \intdf{\psi^2}}{\intf
    \psi^2} = \frac{\Ltf{\del \psi}^2 + \Bi^{-1} A \left(
      \have{\psi^2}|_{z=0} + \have{\psi^2}|_{z=1}
    \right)}{\Ltf{\psi}^2} \geq \lameig_R , 
\end{equation*}
where $\lameig_R > 0$ is the lowest eigenvalue of $-\lap$ on this
domain with homogeneous Robin BCs $\unitn \cdot \del \psi + \Bi^{-1}
\psi = 0$ at $z = 0, 1$, or equivalently $\psi - \Bi \psi_z = 0$ at $z
= 0$, $\psi + \Bi \psi_z = 0$ at $z = 1$ (see \eqref{eq:throbbc}).  In
particular, the temperature perturbation field $\theta$ for our
convection problem with mixed thermal BCs satisfies $\lameig_R \|
\theta\|^2 \leq \| \del \theta\|^2 + \Bi^{-1} A \left(
  \have{\theta^2}|_{z=0} + \have{\theta^2}|_{z=1} \right)$.  However,
since $\theta$ itself satisfies the Robin BCs \eqref{eq:throbbc}, we
have by \eqref{eq:thbdyrobn} that $0 \leq \Bi^{-1} A \left(
  \have{\theta^2}|_{z=0} + \have{\theta^2}|_{z=1} \right) = - A
\zlim{\have{\theta \theta_z}}$; consequently for mixed thermal BCs we
have the inequality
\begin{equation}
  \label{iq:APointhrob}
  - \Ltf{\del \theta}^2 + A \zlim{\have{\theta \theta_z}} \leq -
  \lameig_R \Ltf{\theta}^2 .
\end{equation}

Substituting \eqref{iq:APoinv} and \eqref{iq:APointhrob} into
\eqref{iq:AL2full3rob}, we thus estimate
\begin{align}
  \frac{1}{2} \frac{d}{dt} \left[ \Ltf{\theta}^2 + \frac{32}{\Prandtl
      \R^2} \Ltf{\bv}^2 \right] 
  & \leq \frac{(1+2\Bi)^2}{4\Bi^2} A  - \frac{1}{4} \left[ \lameig_R
    \Ltf{\theta}^2 + \frac{64}{\R^2} \lameig_D \Ltf{\bv}^2 \right]   +
  \frac{3}{4} A  \zlim{\have{\theta \theta_z}}
  \nonumber \\ 
  \label{iq:AL2full4rob}
  & \leq \frac{(1+2\Bi)^2}{4\Bi^2} A  - \frac{\lameig_R}{4} \min
  \left( 2\Prandtl \frac{\lameig_D}{\lameig_R}, 1 \right) \left[
    \Ltf{\theta}^2 + \frac{32}{\Prandtl \R^2} \Ltf{\bv}^2 \right] .
\end{align}

\subsubsection{Uniform boundedness for general thermal boundary
  conditions applied to fluid:}
\label{ssapp:unifbddfluid}

The estimates \eqref{iq:AL2full4dir}, \eqref{iq:AL2full4neu} and
\eqref{iq:AL2full4rob} are all differential inequalities of the form
$dE/dt \leq b_1 - b_2 E$ for constants $b_1 , b_2 > 0$, where $E(t) =
\Ltf{\theta}^2 + (32/\Prandtl \R^2) \Ltf{\bv}^2$.  It follows, using
Gronwall's inequality, that $E(t) \leq E(0) \e^{-b_2 t} + (b_1/b_2) (1
- \e^{-b_2 t})$, thus completing the proof that $\bv$ and $\theta$,
and hence $\bu = \bv$ and $T = \tau + \theta$, are uniformly bounded
(by $E(0) + b_1/b_2$) in $\Ltwo$ for all $t$.
\end{proof}

\section{Governing identities valid in fixed flux limit}
\label{app:ffidentities}

In Section~\ref{ssec:mixedid} the energy identities required to
formulate a bounding principle were stated in a form valid for $0 \leq
\Bi < \infty$, which are appropriate for exploring the fixed
temperature limit $\Bi \to 0$, and let us study the effect of finite
conductivity as a perturbation of the ideal case of perfectly
conducting boundaries.  By instead writing these governing identities
in terms of $\dT$, $\dtau$, $\tave{\have{T^2}|_{z=0,1}}$ and
$\tave{\have{\theta^2}|_{z=0,1}}$, we obtain a formulation valid for
$0 < \Bi \leq \infty$ and relevant to the insulating boundary fixed
flux limit $\Bi \to \infty$.  For completeness, we state these
identities here; of course for $0 < \Bi < \infty$ the formulas
\eqref{eq:gradurobn}--\eqref{eq:bdidrobn} in terms of $\dT$ are
equivalent to \eqref{eq:gradurobd}--\eqref{eq:bdidrobd} in terms of
$\dzT$. 

Solving for $\dzT$ from \eqref{eq:betadelta} and substituting into the
global kinetic energy identity \eqref{eq:gradu2}, we first obtain
\begin{equation}
  \frac{1}{A R} \tave{\Ltf{\del \bu}^2 } = \dzT - \dT = \frac{1}{2\Bi}
  (1 + 2\Bi) \left( 1 - \dT \right) ,
    \label{eq:gradurobn}
\end{equation}
which reduces to \eqref{eq:graduneu} in the limit $\Bi \rtarr \infty$.
Correspondingly, evaluating $\tave{\zlim{\have{T T_z}}}$ by solving
for $T_z$ at $z = 0,1$ in terms of $T$ and averaging, the general form
of the global thermal energy identity \eqref{eq:gradT} which reduces
to the fixed flux limit \eqref{eq:gradTneu} is
\begin{equation}
  \frac{1}{A} \tave{\Ltf{\del T}^2 } = \tave{\zlim{\have{T T_z}}} =
  \dT + \frac{1}{2\Bi} \left( 1 + \dT \right) - \frac{1}{\Bi} \tave{
    \have{T^2}|_{z=0} + \have{T^2}|_{z=1}} . \label{eq:gradTrobn}
\end{equation}

The homogeneous Robin thermal BCs satisfied by the fluctuation
$\theta$ in the background flow formulation (with a background field
$\tau(z)$ obeying the identity \eqref{eq:gammadtau}) may be written
for $\Bi > 0$ as $\unitn \cdot \del \theta + \Bi^{-1} \theta = 0$, or
in our horizontally periodic geometry 
\begin{equation}
  \label{eq:throbbcn}
  \theta_z - \Bi^{-1} \theta = 0 \ \ \text{at}\ \ z = 0, \qquad 
  \theta_z + \Bi^{-1} \theta = 0 \ \ \text{at}\ \ z = 1 ;
\end{equation}
the individual horizontal thermal Fourier modes thus satisfy
\begin{equation}
  \label{eq:thkrobbcn}
  \D \thk(0) = \Bi^{-1} \, \thk(0) , \qquad \D \thk(1) = - \Bi^{-1} \,
  \thk(1) .
\end{equation}
The appropriate formulation of the boundary term in
\eqref{eq:Qtiltbdef} valid for $0 < \Bi \leq \infty$ is thus
$\intdf{\theta \, \unitn \cdot \del \theta} = - \Bi^{-1}
\intdf{\theta^2} \leq 0$, which in our geometry becomes
\begin{align}
  \label{eq:thbdyrobn}
  \zlim{\have{\theta \theta_z}} & = - \frac{1}{\Bi}
  (\have{\theta^2}|_{z=0} + \have{\theta^2}|_{z=1}) \\
  \label{eq:thkbdyrobn}
  & = -\frac{1}{\Bi} \sum_{\bk} \left( |\thk(0)|^2 + |\thk(1)|^2
  \right) \leq 0 
\end{align}
in real and Fourier space, respectively, again verifying the
stabilizing effect of the boundary term in $\Qtiltb$.  Substituting
for $\dzT$ and $\dztau$, we may write the identity \eqref{eq:bgidd}
instead in terms of $\dT$ and $\dtau$, as
\begin{equation}
  \dzT \dtau - \dztau \dT + \dztau \dtau = \frac{1}{2\Bi} (1 + 2\Bi)
  (2 \dtau - \dT ) - \frac{1}{2\Bi} \dtau^2 . \label{eq:bgidn}
\end{equation}
For mixed (Robin) thermal BCs with constant Biot number $\Bi > 0$, we
can now substitute \eqref{eq:gradurobn} and \eqref{eq:bgidn} to write
the governing identity \eqref{eq:bdid2} in terms of $\dT$ as 
\begin{equation}
  \label{eq:bdidrobn}
  \frac{1}{2\Bi} (1 + 2\Bi) (1 - \dT) = \ b \left( \int_0^1 \tau'^2 \,
    dz + \frac{1 + 2\Bi}{2\Bi} ( 1 - 2\dtau) + \frac{1}{2\Bi}
    \dtau^2 \right)  - \frac{b}{A} \Qtiltb[\bv,\theta] 
\end{equation}
(compare \eqref{eq:bdidneu}), using \eqref{eq:thbdyrobn} to evaluate
the boundary term in $\Qtiltb[\bv,\theta]$, which allows us to derive
the lower bound \eqref{iq:deltabnd} on $\dT$; of course for $0 < \Bi <
\infty$, this expression is completely equivalent to
\eqref{eq:bdidrobd}.

\section{Rigorous admissibility conditions and bounds}
\label{app:rigbounds}

In Section~\ref{sec:Bibound}, explicit analytical bounds on the
dependence of the Nusselt number $\Nu$ on the control parameter $\R$
and hence on the Rayleigh number $\Ra$ were obtained, using a piecewise
linear background $\taud(z)$ and elementary estimates on the quadratic
form.  However, for fixed flux and general Biot number thermal BCs
($\Bi > 0$), the conditions \eqref{eq:deltacneu},
\eqref{eq:deltacrob2} on $\delta$ and the bounds \eqref{eq:Nubndneu},
\eqref{eq:Nubndrob} were derived using some asymptotic approximations
for $\delta \ll \min (\Bi,1)$ (that is, for sufficiently large $\R$),
and are thus not \emph{rigorously} applicable for fixed nonzero
$\delta$.  The arguments of Section~\ref{sec:Bibound} may however
readily be adapted to give rigorously valid admissibility conditions
on $\taud$, and hence rigorous bounds on $\Nu(\Ra)$, at the cost of
weakening the prefactors.  

We choose to present the rigorous bounds as results valid uniformly in
$\Bi$ and for $0 \leq \delta \leq 1/2$; the prefactors may be improved
by $\bigO(1)$ corrections in several places by restricting the range
of $\Bi$ and/or $\delta$ values.  In our formal development we account
for all relevant scaling regimes by using a balance parameter $b =
\bproof(\Bi,\delta)$ defined on $\delta \in (0,1/2]$, $\Bi \in
[0,\infty]$ via
\begin{equation}
  \label{eq:bdeltaBi}
  \bproof - 1 =
  \begin{cases}
    \frac{1}{2} & \mathrm{for}\ 0 \leq \Bi \leq \frac{1}{2}, \ \ \delta
    \geq \Bi , \\
    \frac{\delta}{2\Bi} & \mathrm{for}\ 0 < \Bi \leq \frac{1}{2}, \ \
    \delta \leq \Bi , \\
    \frac{1+2\Bi}{4\Bi} \, \delta & \mathrm{for}\ \frac{1}{2} \leq \Bi
    \leq \infty \ \ (\text{so\ necessarily}\ \delta \leq \Bi) .
  \end{cases}
\end{equation}
Note that the function $\bproof(\Bi,\delta)$ is continuous on its
domain and agrees with the values used previously in the limiting
fixed temperature and fixed flux cases, $b_0 = 3/2$ and $b_{\infty} =
1 + \cinf \delta = 1 + \delta/2$, respectively.  We immediately
observe that
\begin{equation}
  \label{iq:brange}
  1 < \bproof \leq \frac{3}{2} ;
\end{equation}
for $\Bi \geq 1/2$ this follows from $(1+2\Bi)/4\Bi = 1/2 + 1/4\Bi
\leq 1$ and $\delta \leq 1/2$.

In the first lemma we estimate the formulas for upper bounds on $\Nu$
and lower bounds on $\dT$ for piecewise linear backgrounds uniformly
in terms of $\delta$ and $\Bi$:
\begin{lemma}
  \label{lem:bndNudelta}
  For general mixed thermal BCs with arbitrary Biot number $\Bi \in
  [0,\infty]$, the upper bound $\Nupwl{\Bi}$ on the Nusselt number and
  lower bound $\BpwldT{\Bi}$ on the temperature drop across the fluid, obtained
  using a piecewise linear background profile of the form $\taud$
  \eqref{eq:taudelta} with $\delta \in (0,1/2]$ and a balance
  parameter $b = \bproof(\Bi,\delta)$ defined by \eqref{eq:bdeltaBi},
  satisfy
  \begin{equation}
    \label{eq:NudTdbnd}
    \Nupwl{\Bi}(\delta,\bproof) \leq \frac{3}{2\delta} , \qquad
    \BpwldT{\Bi}(\delta,\bproof) \geq \frac{\bproof(\Bi,\delta) -
      1}{2} .
  \end{equation}
\end{lemma}

\begin{proof}
  For piecewise linear backgrounds $\taud$, with $\dztau$ and $\dtau$
  given by \eqref{eq:gammaBi}, the upper bounds on $\dzT$ and $\dT$
  from \eqref{eq:dzTrobpwlin}--\eqref{eq:dTrobpwlin} using $b =
  \bproof$ satisfy
  \begin{align*}
    \BpwldzT{\Bi}(\delta,\bproof) & = \dztau + (\bproof -1 )(\dztau - 1)
    = \frac{1}{2(\delta + \Bi)} \left[ 1 + 2\Bi +
      (\bproof-1)(1-2\delta) \right] , \\ 
    \BpwldT{\Bi}(\delta,\bproof) & = \dtau + (\bproof - 1)(\dtau - 1)
    = \frac{\delta}{\delta+\Bi} \left[ 1 + 2\Bi + (\bproof -
      1)\frac{\Bi}{\delta} (2\delta - 1) \right], 
  \end{align*}
  so that 
  \begin{equation*}
    \Nupwl{\Bi}(\delta,\bproof) =
    \frac{\BpwldzT{\Bi}(\delta,\bproof)}{\BpwldT{\Bi}(\delta,\bproof)}
    =  \frac{1}{2\delta} \, \frac{1+2\Bi +
      (\bproof-1)(1-2\delta)}{1+2\Bi + (\bproof-1)(2\delta-1)\Bi/\delta} .
  \end{equation*}
  Now we consider the three cases in the definition of $\bproof$:
  \begin{itemize}
  \item[I:] For $0 \leq \Bi \leq \delta \leq 1/2$:\\
    We have $\bproof - 1 = 1/2$, so $1 + 2\Bi + (\bproof-1)(1-2\delta)
    = 3/2 + 2\Bi - \delta \leq 3/2 + 2\Bi$ and $1+2\Bi +
    (\bproof-1)(2\delta-1)\Bi/\delta = 1 + 3\Bi - \Bi/2\delta \geq 1/2
    + 3\Bi$, so that
    \begin{equation*}
      \Nupwl{\Bi}(\delta,\bproof) \leq \frac{1}{2\delta} \,
      \frac{3+4\Bi}{1+6\Bi} \leq \frac{3}{2\delta} , \qquad
      \BpwldT{\Bi}(\delta,\bproof) \geq \frac{\delta}{\delta+\Bi}
      \left[ \frac{1}{2} + 3\Bi \right] \geq \frac{1}{2} \,
      \frac{\delta}{\delta+\Bi} \geq \frac{1}{4} .
    \end{equation*}

  \item[II:] For $0 < \delta \leq \Bi \leq 1/2$:\\
    We have $\bproof - 1 = \delta/2\Bi$, so $1 + 2\Bi +
    (\bproof-1)(1-2\delta) \leq 1 + 2\Bi + \delta/2\Bi \leq 3/2 +
    2\Bi$ and $1+2\Bi + (\bproof-1)(2\delta-1)\Bi/\delta = 1/2 + 2\Bi +
    \delta \geq 1/2 + 2\Bi$, so that
    \begin{equation*}
      \Nupwl{\Bi}(\delta,\bproof) \leq \frac{1}{2\delta} \,
      \frac{3+4\Bi}{1+4\Bi} \leq \frac{3}{2\delta} , \qquad 
      \BpwldT{\Bi}(\delta,\bproof) \geq \frac{\delta}{\delta+\Bi}
      \left[ \frac{1}{2} + 2\Bi \right] \geq \frac{1}{2} \,
      \frac{\delta}{\delta+\Bi} \geq \frac{\delta}{4\Bi} .
    \end{equation*}

  \item[III:] For $0 < \delta \leq 1/2$, $\Bi \geq 1/2$:\\
    We have $\bproof - 1 = (1+2\Bi)\delta/4\Bi = \cBi \delta$ (see
    \eqref{eq:cBi}), so $1 + 2\Bi + (\bproof-1)(1-2\delta) = (1+2\Bi)
    [1 + \delta/4\Bi (1-2\delta)] \leq (1 + 2\Bi)[1 + \delta/4\Bi]
    \leq (1+2\Bi) 5/4$ and $1+2\Bi + (\bproof-1)(2\delta-1)\Bi/\delta
    = (1 + 2\Bi) [1 + (2\delta-1)/4 ] = (1+2\Bi) [3/4 + \delta/2] \geq
    (1+2\Bi) 3/4$, so that
    \begin{equation}
      \label{iq:NudTpwllargeBi}
      \Nupwl{\Bi}(\delta,\bproof) \leq \frac{1}{2\delta} \,
      \frac{4+\delta/\Bi}{3+2\delta} \leq \frac{5}{6\delta} , \qquad
      \BpwldT{\Bi}(\delta,\bproof) \geq \frac{\delta}{\delta+\Bi}
      \frac{3}{4} (1+2\Bi) \geq \frac{3}{8} \frac{1+2\Bi}{\Bi} \delta
      . 
    \end{equation}
  \end{itemize}
\end{proof}

\begin{remark}
  The form of the bounds in \eqref{eq:NudTdbnd} is chosen so that
  $\Nupwl{\Bi}$ is uniform in $\Bi$, and $\BpwldT{\Bi}$ is continuous
  on $\Bi \in [0,\infty]$, $\delta \in (0,1/2]$; however, the
  prefactors $C_i$ in the bounds of the form $\Nupwl{\Bi} \leq
  C_1/\delta$, $\BpwldT{\Bi} \geq C_2(\bproof - 1)$ may certainly be
  improved for particular values or over restricted ranges of $\Bi$ by
  $\bigO(1)$ factors from their values $C_1 = 3/2$, $C_2 = 1/2$ in
  \eqref{eq:NudTdbnd}.

  For instance, for fixed temperature BCs $\Bi = 0$, from
  \eqref{eq:Nudpwldir} we have $\Nupwl{0} \leq 3/4\delta$ and
  $\BpwldT{0} = 1$; while for fixed flux BCs $\Bi = \infty$ we can
  show from \eqref{eq:Nudpwlneu} that $\Nupwl{\infty} \leq 2/3\delta$
  and $\BpwldT{\infty} \geq 3\delta/2$.  Furthermore, case III in the
  proof of the above Lemma shows clearly in \eqref{iq:NudTpwllargeBi}
  that for $\Bi \geq 1/2$, it is sufficient to take $C_1 = 5/6$, $C_2
  = 3/2$.
\end{remark}

\vspace{1ex}

Lemma~\ref{lem:bndNudelta} gives an upper bound $\Nupwl{\Bi}$ on $\Nu$
and a lower bound $\BpwldT{\Bi}$ on $\dT$ in terms of $\delta$,
\emph{provided} $\delta$ is such that the piecewise linear background
$\taud$ is admissible for $\Reff = b\R/(b-1)$ where $b = \bproof$ is
defined in \eqref{eq:bdeltaBi}.  The next result gives sufficient
conditions on $\delta$, of the form $\delta \leq \dsuff$, for
admissibility of $\taud$ as a function of $\R$ and $\Bi$, using a
balance parameter $\badm$ which coincides, at $\delta = \dsuff$, with
$\bproof$ used in Lemma~\ref{lem:bndNudelta}:
\begin{lemma}
  \label{lem:deltasuff}
  Consider Rayleigh-B\'enard convection subject to thermal boundary conditions
  with Biot number $\Bi \in [0,\infty]$.
  \begin{itemize}
  \item[(a)] For each $\Reff > 0$, the piecewise linear background
    field $\taud(z)$ defined in \eqref{eq:taudelta} is admissible,
    $\Qtiltb [\bv,\theta] \geq 0$ for all allowed fields $\bv$ and
    $\theta$ (see Section~\ref{ssec:admiss}), provided $\delta \leq
    \min (\delta_c,1/2)$, where $\delta_c$ satisfies
    \eqref{eq:deltacrob}:
    \begin{equation*}
      \delta_c^4 = 32 \frac{(\delta_c+\Bi)^2}{(1+2\Bi)^2} \Reff^{-1} .
    \end{equation*}
    In particular, for any $\R > 0$ and $b > 1$ we may choose $\Reff =
    b \R/(b-1)$.

  \item[(b)] For each $\Bi \geq 0$ and $\R > 0$, a sufficient
    condition for the piecewise linear background $\taud$ to be
    admissible for $\Reff = \badm \R/(\badm-1)$ is that $\delta \leq
    \min (\dsuff,1/2)$, where $\badm(\Bi,\R)$ and $\dsuff(\Bi,\R)$ are
    defined as follows:
    \begin{itemize}
    \item[I:] For $0 \leq \Bi \leq 1/2$ and $\R < \Rsuff(\Bi) = (8/3)
      \Bi^{-2}$ (where we define $\Rsuff(0) = \infty$):
      \begin{equation*}
        \badm = \frac{3}{2}, \qquad \dsuff = 2 \left( \frac{2}{3}
      \right)^{1/2} \R^{-1/2} ;
      \end{equation*}
    \item[II:] For $0 < \Bi \leq 1/2$ and $\R \geq \Rsuff(\Bi) = (8/3) \Bi^{-2}$:
      \begin{equation*}
        \badm = 1 + \frac{\dsuff}{2\Bi}, \qquad \dsuff = 2 \left(
          \frac{\Bi}{3} \right)^{1/3} \R^{-1/3} ;
      \end{equation*}
    \item[III:] For $1/2 < \Bi \leq \infty$ and all $\R$:
      \begin{equation*}
        \badm = 1 + \frac{1+2\Bi}{4\Bi} \dsuff, \qquad \dsuff = 2
        \left( \frac{2}{3} \right)^{1/3} \left( \frac{\Bi}{1+2\Bi}
        \right)^{1/3} \R^{-1/3} .
      \end{equation*}
    \end{itemize}
    Furthermore, for $0 < \Bi \leq 1/2$ we have that $\dsuff(\Bi,\R)$
    is continuous in $\R$, and $\dsuff \geq \Bi$ if and only if $\R
    \leq \Rsuff(\Bi)$.
  \end{itemize}
\end{lemma}

\begin{proof}
  Part (a) of the lemma was proved in Section~\ref{ssec:CSestimates}
  using Cauchy-Schwarz estimates on the indefinite term in $\Qk$,
  where we also used the relationship from Section~\ref{ssec:admiss}
  between admissibility ($\Qtiltb \geq 0$) and positivity of the
  quadratic forms in Fourier space, $\Qk \geq 0$ for all $\bk$.

  To prove part (b), we first observe for $0 < \Bi \leq 1/2$ that for $\R
  = \Rsuff(\Bi)$, we have $2 (2/3)^{1/2} \Rsuff^{-1/2} = \Bi = 2
  (\Bi/3)^{1/3} \Rsuff^{-1/3}$; hence $\dsuff(\Bi,\R)$ is a
  continuous function of $\R$, and $\dsuff \geq \Bi$ when $\R \leq
  \Rsuff(\Bi)$, while for $\R \geq \Rsuff(\Bi)$ we
  have $\dsuff \leq \Bi$.  The function $\badm$ thus agrees with
  $\bproof(\Bi,\dsuff)$ from \eqref{eq:bdeltaBi}, and as in
  \eqref{iq:brange} we have that $1 < \badm \leq 3/2$.

  Next, we define $\Reffs(\Bi,\R) = \badm \R/(\badm - 1)$ as in the
  statement of the lemma, and then let $\delta_c = \delta_c(\Bi,\R)$
  be the critical value of $\delta$ as in part (a) for this $\Reff =
  \Reffs(\Bi,\R)$.  Then since the result of (a) indicates that
  $\taud$ is admissible for this $\Reff$ whenever $\delta \leq \min
  (\delta_c(\Bi,\R), 1/2)$, to conclude the proof of part (b) it is
  sufficient to show that $\dsuff \leq \delta_c$ in each of the three
  cases:
  \begin{itemize}
  \item[I:] For $0 \leq \Bi \leq 1/2$ and $\R < \Rsuff(\Bi)$:
    \begin{equation*}
      \delta_c^4 = 32 \frac{(\delta_c+\Bi)^2}{(1+2\Bi)^2}
      \frac{1/2}{3/2} \R^{-1} \geq 32 \frac{\delta_c^2}{4} \frac{1}{3} 
      \R^{-1} = \frac{8}{3} \R^{-1} \, \delta_c^2 = \dsuff^2
      \delta_c^2 .
    \end{equation*}
  \item[II:] For $0 < \Bi \leq 1/2$ and $\R \geq \Rsuff(\Bi)$:
    \begin{equation*}
      \delta_c^4 = 32 \frac{(\delta_c+\Bi)^2}{(1+2\Bi)^2}
      \frac{\dsuff/2\Bi}{\badm} \R^{-1} \geq 32 \frac{\Bi^2}{4}
      \frac{\dsuff}{2\Bi} \frac{2}{3} \R^{-1} = \frac{8}{3} \Bi
      \R^{-1} \, \dsuff = \dsuff^4 .
    \end{equation*}
  \item[III:] For $1/2 < \Bi \leq \infty$ and $\R > 0$:
    \begin{equation*}
      \delta_c^4 = 32 \frac{(\delta_c+\Bi)^2}{(1+2\Bi)^2}
      \frac{(1+2\Bi)\, \dsuff}{4\Bi} \frac{1}{\badm} \R^{-1} \geq 8
      \frac{\Bi^2}{(1+2\Bi)^2} \frac{(1+2\Bi)\, \dsuff}{\Bi} \frac{2}{3}
      \R^{-1} = \frac{16}{3} \frac{\Bi}{1+2\Bi} \R^{-1} \, \dsuff =
      \dsuff^4 .
    \end{equation*}
  \end{itemize}
\end{proof}

We may now combine the above results to obtain rigorous and uniformly
valid analytical bounds on $\Nu(\Ra)$ for convection with general Biot
number thermal BCs.

\begin{thm}
  \label{thm:NuRarig}
  For Rayleigh-B\'enard convection subject to thermal boundary
  conditions with Biot number $\Bi \in [0,\infty]$, let the control
  parameter $\R$ satisfy
  \begin{equation*}
    \R \geq \frac{32}{3} \ \ \ \text{for}\ \ \ \Bi \leq \frac{1}{2},
    \qquad \R \geq \frac{128}{3} \frac{\Bi}{1+2\Bi} \ \ \ \text{for}\ \ \
    \Bi  \geq \frac{1}{2} .  
  \end{equation*}
  Then the Nusselt number is bounded according to
  \begin{equation}
    \label{iq:NuRabndrig}
    \Nu \leq \frac{3\sqrt{6}}{4} \Ra^{1/2} ,
  \end{equation}
  independently of $\Bi$.
\end{thm}

\begin{proof}
  We fix a Biot number $\Bi \in [0,\infty]$ and control parameter $\R
  > 0$, and define $\dsuff(\Bi,\R)$, $\badm(\Bi,\R)$ and $\Reff =
  \badm \R/(\badm - 1)$ as in Lemma~\ref{lem:deltasuff}; the
  constraint on $\R$ ensures that $\dsuff(\Bi,\R) \leq 1/2$.  Then
  using $\delta = \dsuff$, we apply the background method with the
  piecewise linear background $\tau_{\dsuff}$, which is admissible by
  Lemma~\ref{lem:deltasuff}(b), so that by
  \eqref{iq:betabnd}--\eqref{iq:Nubnd}, $\Nupwl{\Bi}(\dsuff,\badm)$
  and $\BpwldT{\Bi}(\dsuff,\badm)$ are rigorous bounds (upper and
  lower, respectively) for $\Nu$ and $\dT$.  From this theorem it also
  follows that $\badm$ coincides with the balance parameter $\bproof$
  \eqref{eq:bdeltaBi}, $\badm(\Bi,\R) = \bproof(\Bi,\dsuff(\Bi,\R))$,
  so that we may apply the results of Lemma~\ref{lem:bndNudelta}, and
  conclude that $\Nu \leq \Nupwl{\Bi}(\dsuff,\badm) \leq 3/2\dsuff$
  and $\dT \geq \BpwldT{\Bi}(\dsuff,\badm) \geq (\badm - 1)/2$, from
  which we may deduce a lower bound on $\Ra = \R \dT$.

  In evaluating \eqref{eq:NudTdbnd} at $\delta = \dsuff$, we consider
  the three cases in the definition of $\badm(\Bi,\R)$ and
  $\dsuff(\Bi,\R)$:
  \begin{itemize}
  \item[I:] For $0 \leq \Bi \leq 1/2$ and $\R < \Rsuff(\Bi) = (8/3)
    \Bi^{-2}$, in which case $\badm - 1 = 1/2$ and $\dsuff = 2(2/3)^{1/2}
    \R^{-1/2} > \Bi$: 
    \begin{equation*}
      \Nu \leq \frac{3}{2\dsuff} = \frac{3\sqrt{6}}{8} \R^{1/2} ,
      \qquad \Ra = \R \dT \geq \frac{1}{4} \R ,
    \end{equation*}
    and so $\Nu \leq \frac{3}{8}\sqrt{6} (4\Ra)^{1/2} = \frac{3}{4}
    \sqrt{6}\, \Ra^{1/2}$.

  \item[II:] For $0 < \Bi \leq 1/2$ and $\R \geq \Rsuff(\Bi)$, for
    which $\badm - 1 = \dsuff/2\Bi$ and $\dsuff = 2 (\Bi/3)^{1/3}
    \R^{-1/3}$: 
    \begin{equation*}
      \Nu \leq \frac{3}{2\dsuff} = \frac{3}{4} \left( \frac{3}{\Bi}
      \right)^{1/3} \R^{1/3} , \qquad \Ra = \R \dT \geq
      \frac{\dsuff}{4\Bi} \R = \frac{1}{2} \left( \frac{1}{3\Bi^2}
      \right)^{1/3} \R^{2/3} ,
    \end{equation*}
    which gives $\R^{1/3} \leq 2^{1/2} 3^{1/6} \Bi^{1/3} \Ra^{1/2}$
    and thus $\Nu \leq \frac{3}{4} \sqrt{6}\, \Ra^{1/2}$.

  \item[III:] For $1/2 < \Bi \leq \infty$ and $\R > 0$, for which
    $\badm - 1 = (1+2\Bi)\dsuff/4\Bi$ and $\dsuff = 2 ( 2/3)^{1/3} (
    \Bi/(1+2\Bi) )^{1/3} \R^{-1/3}$:
    \begin{equation*}
      \Nu \leq \frac{3}{2\dsuff} = \frac{3}{4} \left( \frac{3}{2}
        \frac{1+2\Bi}{\Bi} \right)^{1/3} \R^{1/3} , \qquad \Ra = \R
      \dT \geq \frac{1}{4} \left( \frac{2}{3} \right)^{1/3} \left(
        \frac{1+2\Bi}{\Bi} \right)^{2/3} \R^{2/3} ,
    \end{equation*}
    and so $\R^{1/3} \leq 2^{5/6} 3^{1/6} (\Bi/(1+2\Bi))^{1/3}
    \Ra^{1/2}$ and $\Nu \leq \frac{3}{4} \sqrt{6}\, \Ra^{1/2}$.
  \end{itemize}
\end{proof}

\begin{remark}
  The upper bound \eqref{iq:NuRabndrig} $\Nu \leq C_{\text{rig,unif}}
  \Ra^{1/2}$ on $\Nu(\Ra)$, valid uniformly in $\R$ (and hence $\Ra$)
  and in $\Bi \in [0,\infty]$, was obtained at the cost of weakening
  the prefactor $C_{\text{rig,unif}} = 3\sqrt{6}/4 \approx 1.837$
  relative to the asymptotic scaling \eqref{eq:Nubndrob} $\Nu \leq
  C_{\text{asym}} \Ra^{1/2}$, $C_{\text{asym}} = \sqrt{2/27} \approx
  0.272$ for $\Bi > 0$.  As in Lemma~\ref{lem:bndNudelta}, the
  prefactor $C = C_{\text{rig}}$ in the rigorous bound $\Nu \leq C
  \Ra^{1/2}$ may be improved for particular (intervals of) $\Bi$.  For
  instance, for fixed temperature BCs ($\Bi = 0$), \eqref{eq:Nubnddir}
  establishes the bound with $C = 3\sqrt{6}/32 \approx 0.230$; for
  fixed flux BCs ($\Bi = \infty$) one can prove that it is sufficient
  to take $C = \sqrt{5/54} \approx 0.304$; and uniformly for $1/2 \leq
  \Bi \leq \infty$, the use of the estimates \eqref{iq:NudTpwllargeBi}
  from Lemma~\ref{lem:bndNudelta} instead of \eqref{eq:NudTdbnd}
  immediately allows one to improve the prefactor in
  \eqref{iq:NuRabndrig} to $C = 5\sqrt{2}/12 \approx 0.589$.
\end{remark}



\begin{thebibliography}{45}
\expandafter\ifx\csname natexlab\endcsname\relax\def\natexlab#1{#1}\fi

\bibitem[Amati {\em et~al.\/}(2005)Amati, Koal, Massaioli, Sreenivasan \&
  Verzicco]{AKMSV05}
{\sc Amati, G., Koal, K., Massaioli, F., Sreenivasan, K.~R. \& Verzicco, R.}
  2005 Turbulent thermal convection at high {R}ayleigh numbers for a
  {B}oussinesq fluid of constant {P}randtl number. {\em Phys. Fluids\/} {\bf
  17}, 121701.

\bibitem[Brown {\em et~al.\/}(2005)Brown, Nikolaenko, Funfschilling \&
  Ahlers]{BNFA05}
{\sc Brown, E., Nikolaenko, A., Funfschilling, D. \& Ahlers, G.} 2005 Heat
  transport in turbulent {R}ayleigh-b\'enard convection: {E}ffect of finite
  top- and bottom-plate conductivities. {\em Phys. Fluids\/} {\bf 17}, 075108.

\bibitem[Busse(1969)]{Buss69}
{\sc Busse, F.~H.} 1969 On {H}oward's upper bound for heat transport by
  turbulent convection. {\em J. Fluid Mech.\/} {\bf 37}, 457--477.

\bibitem[Busse \& Riahi(1980)]{BuRi80}
{\sc Busse, F.~H. \& Riahi, N.} 1980 Nonlinear convection in a layer with
  nearly insulating boundaries. {\em J. Fluid Mech.\/} {\bf 96}, 243--256.

\bibitem[Chan(1971)]{Chan71}
{\sc Chan, S.-K.} 1971 Infinite {P}randtl number convection. {\em Stud. Appl.
  Math.\/} {\bf 50}, 13--49.

\bibitem[Chapman \& Proctor(1980)]{ChPr80}
{\sc Chapman, C.~J. \& Proctor, M. R.~E.} 1980 Nonlinear {R}ayleigh-{B}\'enard
  convection between poorly conducting boundaries. {\em J. Fluid Mech.\/} {\bf
  101}~(4), 759--782.

\bibitem[Chaumat {\em et~al.\/}(2002)Chaumat, Castaing \& Chill\`a]{CCC02}
{\sc Chaumat, S., Castaing, B. \& Chill\`a, F.} 2002 Rayleigh-{B}\'enard cells:
  influence of the plates' properties. In {\em Advances in Turbulence IX,
  Proceedings of the Ninth European Turbulence Conference\/} (ed. I.~P. Castro,
  P.~E. Hancock \& T.~G. Thomas), pp. 159--162. CIMNe, Barcelona.

\bibitem[Chavanne {\em et~al.\/}(1997)Chavanne, Chill\`a, Castaing, H\'ebral,
  Chabaud \& Chaussy]{CCCHCC97}
{\sc Chavanne, X., Chill\`a, F., Castaing, B., H\'ebral, B., Chabaud, B. \&
  Chaussy, J.} 1997 Observation of the ultimate regime in {R}ayleigh-{B}\'enard
  convection. {\em Phys. Rev. Lett.\/} {\bf 79}, 3648--3651.

\bibitem[Chavanne {\em et~al.\/}(2001)Chavanne, Chill\'a, Chabaud, Castaing \&
  H\'ebral]{CCCCH01}
{\sc Chavanne, X., Chill\'a, F., Chabaud, B., Castaing, B. \& H\'ebral, B.}
  2001 Turbulent {R}ayleigh-{B}\'enard convection in gaseous and liquid {He}.
  {\em Phys. Fluids\/} {\bf 13}, 1300--1320.

\bibitem[Chill\`a {\em et~al.\/}(2004)Chill\`a, Rastello, Chaumat \&
  Castaing]{CRCC04}
{\sc Chill\`a, F., Rastello, M., Chaumat, S. \& Castaing, B.} 2004 Ultimate
  regime in {R}ayleigh-{B}\'enard convection: {T}he role of plates. {\em Phys.
  Fluids\/} {\bf 16}, 2452--2456.

\bibitem[Constantin \& Doering(1996)]{CoDo96}
{\sc Constantin, P. \& Doering, C.~R.} 1996 Heat transfer in convective
  turbulence. {\em Nonlinearity\/} {\bf 9}, 1049--1060.

\bibitem[Constantin \& Doering(1999)]{CoDo99}
{\sc Constantin, P. \& Doering, C.~R.} 1999 Infinite {P}randtl number
  convection. {\em J. Stat. Phys.\/} {\bf 94}~(1/2), 159--172.

\bibitem[Cross \& Hohenberg(1993)]{CrHo93}
{\sc Cross, M. \& Hohenberg, P.} 1993 Pattern formation outside of equilibrium.
  {\em Rev. Mod. Phys.\/} {\bf 65}~(3), 851--1112.

\bibitem[Doering \& Constantin(1992)]{DoCo92}
{\sc Doering, C.~R. \& Constantin, P.} 1992 Energy dissipation in shear driven
  turbulence. {\em Phys. Rev. Lett.\/} {\bf 69}~(11), 1648--1651.

\bibitem[Doering \& Constantin(1996)]{DoCo96}
{\sc Doering, C.~R. \& Constantin, P.} 1996 Variational bounds on energy
  dissipation in incompressible flows. {III}. {C}onvection. {\em Phys. Rev.
  E\/} {\bf 53}~(6), 5957--5981.

\bibitem[Doering {\em et~al.\/}(2006)Doering, Otto \& Reznikoff]{DOR06}
{\sc Doering, C.~R., Otto, F. \& Reznikoff, M.~G.} 2006 Bounds on vertical heat
  transport for infinite-{P}randtl-number {R}ayleigh-{B}\'enard convection.
  {\em J. Fluid Mech.\/} {\bf 560}, 229--241.

\bibitem[Gao(2006)]{Gao06m}
{\sc Gao, J.} 2006 Upper bounds on the heat transport in {R}ayleigh-{B}\'enard
  convection. Master's thesis, Simon Fraser University.

\bibitem[Glazier {\em et~al.\/}(1999)Glazier, Segawa, Naert \& Sano]{GSNS99}
{\sc Glazier, J.~A., Segawa, T., Naert, A. \& Sano, M.} 1999 Evidence against
  `ultrahard' thermal turbulence at very high {R}ayleigh numbers. {\em
  Nature\/} {\bf 398}, 307--310.

\bibitem[Grossmann \& Lohse(2001)]{GrLo01}
{\sc Grossmann, S. \& Lohse, D.} 2001 Thermal convection for large {P}randtl
  numbers. {\em Phys. Rev. Lett.\/} {\bf 86}~(15), 3316--3319.

\bibitem[Hopf(1941)]{Hopf41}
{\sc Hopf, E.} 1941 Ein allgemeiner {E}ndlichkeitssatz der {H}ydrodynamik. {\em
  Math. Ann.\/} {\bf 117}, 764--775.

\bibitem[Howard(1963)]{Howa63}
{\sc Howard, L.~N.} 1963 Heat transport by turbulent convection. {\em J. Fluid
  Mech.\/} {\bf 17}, 405--432.

\bibitem[Hurle {\em et~al.\/}(1967)Hurle, Jakeman \& Pike]{HJP67}
{\sc Hurle, D. T.~J., Jakeman, E. \& Pike, E.~R.} 1967 On the solution of the
  {B}enard problem with boundaries of finite conductivity. {\em Proc. R. Soc.
  London A\/} {\bf 296}~(1447), 469--475.

\bibitem[Ierley {\em et~al.\/}(2006)Ierley, Kerswell \& Plasting]{IKP06}
{\sc Ierley, G.~R., Kerswell, R.~R. \& Plasting, S.~C.} 2006
  Infinite-{P}randtl-number convection. {P}art 2. {A} singular limit of upper
  bound theory. {\em J. Fluid Mech.\/} {\bf 560}, 159--227.

\bibitem[Johnston \& Doering(2007)]{JoDo07}
{\sc Johnston, H. \& Doering, C.~R.} 2007 Rayleigh-{B}\'enard convection with
  imposed heat flux. {\em Chaos\/} {\bf 17}, 041103.

\bibitem[Johnston \& Doering(2008)]{JoDo08}
{\sc Johnston, H. \& Doering, C.~R.} 2008 A comparison of turbulent thermal
  convection between conditions of constant temperature and constant flux.
  Preprint.

\bibitem[Kadanoff(2001)]{Kada01}
{\sc Kadanoff, L.~P.} 2001 Turbulent heat flow: Structures and scaling. {\em
  Physics Today\/} pp. 34--39.

\bibitem[Kerswell(1997)]{Kers97}
{\sc Kerswell, R.~R.} 1997 Variational bounds on shear-driven turbulence and
  turbulent {B}oussinesq convection. {\em Physica D\/} {\bf 100}, 355--376.

\bibitem[Kerswell(2001)]{Kers01}
{\sc Kerswell, R.~R.} 2001 New results in the variational approach to turbulent
  {B}oussinesq convection. {\em Phys. Fluids\/} {\bf 13}~(1), 192--209.

\bibitem[Kraichnan(1962)]{Krai62}
{\sc Kraichnan, R.~H.} 1962 Turbulent thermal convection at arbitrary {P}randtl
  number. {\em Phys. Fluids\/} {\bf 5}, 1374--1389.

\bibitem[Malkus(1954)]{Malk54}
{\sc Malkus, M. V.~R.} 1954 The heat transport and spectrum of thermal
  turbulence. {\em Proc. Roy. Soc. Lond. A\/} {\bf 225}, 196--212.

\bibitem[Nicodemus {\em et~al.\/}(1997)Nicodemus, Grossmann \& Holthaus]{NGH97}
{\sc Nicodemus, R., Grossmann, S. \& Holthaus, M.} 1997 Improved variational
  principle for bounds on energy dissipation in turbulent shear flow. {\em
  Physica D\/} {\bf 101}, 178--190.

\bibitem[Niemela {\em et~al.\/}(2000)Niemela, Skrbek, Sreenivasan \&
  Donnelly]{NSSD00}
{\sc Niemela, J.~J., Skrbek, L., Sreenivasan, K.~R. \& Donnelly, R.~J.} 2000
  Turbulent convection at very high {R}ayleigh numbers. {\em Nature\/} {\bf
  404}, 837--840.

\bibitem[Niemela \& Sreenivasan(2006{\natexlab{{\em a\/}}})]{NiSr06b}
{\sc Niemela, J.~J. \& Sreenivasan, K.~R.} 2006{\natexlab{{\em a\/}}} Turbulent
  convection at high {R}ayleigh numbers and aspect ratio 4. {\em J. Fluid
  Mech.\/} {\bf 557}, 411--422.

\bibitem[Niemela \& Sreenivasan(2006{\natexlab{{\em b\/}}})]{NiSr06}
{\sc Niemela, J.~J. \& Sreenivasan, K.~R.} 2006{\natexlab{{\em b\/}}} The use
  of cryogenic helium for classical turbulence: {P}romises and hurdles. {\em J.
  Low Temp. Phys.\/} {\bf 143}, 163--212.

\bibitem[Otero {\em et~al.\/}(2002)Otero, Wittenberg, Worthing \&
  Doering]{OWWD02}
{\sc Otero, J., Wittenberg, R.~W., Worthing, R.~A. \& Doering, C.~R.} 2002
  Bounds on {R}ayleigh-{B}\'enard convection with an imposed heat flux. {\em J.
  Fluid Mech.\/} {\bf 473}, 191--199.

\bibitem[Plasting \& Kerswell(2003)]{PlKe03}
{\sc Plasting, S.~C. \& Kerswell, R.~R.} 2003 Improved upper bound on the
  energy dissipation rate in plane {C}ouette flow: the full solution to
  {B}usse's problem and the {C}onstantin-{D}oering-{H}opf problem with
  one-dimensional background field. {\em J. Fluid Mech.\/} {\bf 477}, 363--379.

\bibitem[Siggers {\em et~al.\/}(2004)Siggers, Kerswell \& Balmforth]{SKB04}
{\sc Siggers, J.~H., Kerswell, R.~R. \& Balmforth, N.~J.} 2004 Bounds on
  horizontal convection. {\em J. Fluid Mech.\/} {\bf 517}, 55--70.

\bibitem[Sommeria(1999)]{Somm99}
{\sc Sommeria, J.} 1999 The elusive `ultimate state' of thermal convection.
  {\em Nature\/} {\bf 398}, 294--295.

\bibitem[Sparrow {\em et~al.\/}(1964)Sparrow, Goldstein \& Jonsson]{SGJ64}
{\sc Sparrow, E.~M., Goldstein, R.~J. \& Jonsson, V.~K.} 1964 Thermal
  instability in a horizontal fluid layer: effect of boundary conditions and
  non-linear temperature profile. {\em J. Fluid Mech.\/} {\bf 18}, 513--528.

\bibitem[Verzicco(2004)]{Verz04}
{\sc Verzicco, R.} 2004 Effects of nonperfect thermal sources in turbulent
  thermal convection. {\em Phys. Fluids\/} {\bf 16}, 1965--1979.

\bibitem[Verzicco \& Sreenivasan(2008)]{VeSr08}
{\sc Verzicco, R. \& Sreenivasan, K.~R.} 2008 A comparison of turbulent thermal
  convection between conditions of constant temperature and constant heat flux.
  {\em J. Fluid Mech.\/} {\bf 595}, 203--219.

\bibitem[Wang(2008)]{Wang08b}
{\sc Wang, X.} 2008 Bound on vertical heat transport at large {P}randtl number.
  {\em Physica D\/} {\bf 237}, 854--858.

\bibitem[Wei(2007)]{Wei07}
{\sc Wei, Q.} 2007 Bounds on natural convective heat transfer in a porous layer
  with fixed heat flux. {\em Int. Comm. Heat Mass Transfer\/} {\bf 34},
  456--463.

\bibitem[Wittenberg(2008)]{Witt08ub}
{\sc Wittenberg, R.~W.} 2008 Bounds on {R}ayleigh-{B}\'enard convection with
  general thermal boundary conditions. {P}art 2. {I}mperfectly conducting
  plates. Preprint, submitted to \textit{J. Fluid Mech.}

\bibitem[Wittenberg \& Gao(2008)]{WiGa08prep}
{\sc Wittenberg, R.~W. \& Gao, J.} 2008 Conservative bounds on
  {R}ayleigh-{B}\'enard convection with mixed thermal boundary conditions. In
  preparation.

\end{thebibliography}
\end{document}